\documentclass[12pt, a4paper]{article}

\usepackage[margin=1.25in]{geometry}
\usepackage{amsmath,amssymb,amsthm,mathtools}
\usepackage{natbib}
\usepackage{graphicx}
\usepackage[nodisplayskipstretch]{setspace} %
\usepackage[inline]{enumitem}
\usepackage{fullpage,rotating,afterpage}
\usepackage{microtype} %
\usepackage{longtable}
\usepackage{subfigure}
\usepackage{ragged2e}

\usepackage{siunitx}
\usepackage{array}
\usepackage{hyperref}
\hypersetup{
	colorlinks   = true, 
	urlcolor     = blue,
	linkcolor    = black, 
	citecolor   =  black 
}

\usepackage[T1]{fontenc} 
\usepackage{tikz}
\usepackage{pdflscape}
\usepackage{booktabs}
\usepackage[justification=centering]{caption}
\usepackage[flushleft]{threeparttable}
\usepackage{multirow}
\usepackage{bm}
\usepackage{xcolor}

\usepackage{footnote}

\usepackage{xr}

\usepackage{algorithm}
\usepackage{algpseudocode}

\newcommand{\E}{\mathbb{E}}

\newcommand{\IR}{\mathbb{R}}

\newcommand{\1}[1]{\mathbf{1}\{#1\}}
\newcommand{\onepm}[1]{\mathbf{1}^\pm_{\{#1\}}} 
\newcommand{\Var}{\operatorname{Var}}

\newcommand{\Cov}{\operatorname{Cov}}

\newcommand{\wh}{\widehat}
\newcommand{\wt}{\widetilde}
\newcommand{\mc}{\mathcal}

\newcommand{\sss}{\scriptscriptstyle}

\DeclareMathOperator*{\argmin}{argmin}

\newtheorem{proposition}{Proposition} 

\newtheorem{theorem}{Theorem}

\newtheorem{lemma}{Lemma}
\newtheorem{assumption}{Assumption}
\theoremstyle{definition}

\newtheorem{remark}{Remark}

\newtheorem{assumptionalt}{Assumption}[assumption]
\newenvironment{ass}[1]{
	\renewcommand\theassumptionalt{#1}
	\assumptionalt
}{\endassumptionalt}

\newtheorem{exampleA}{Example}

\newtheorem{exampleB}{Example}

\numberwithin{equation}{section}

\usepackage{verbatim,color}

\def\boxit#1{\vbox{\hrule\hbox{\vrule\kern6pt
			\vbox{\kern6pt#1\kern6pt}\kern6pt\vrule}\hrule}}

\usepackage{dcolumn} 
\newcolumntype{d}[1]{D{.}{.}{#1}}

\newcommand{\captionfonts}{\small}

\makeatletter  
\long\def\@makecaption#1#2{%
	\vskip\abovecaptionskip
	\sbox\@tempboxa{{\captionfonts #1: #2}}%
	\ifdim \wd\@tempboxa >\hsize
	{\captionfonts #1: #2\par}
	\else
	\hbox to\hsize{\hfil\box\@tempboxa\hfil}%
	\fi
	\vskip\belowcaptionskip}
\makeatother   

\usepackage{titlesec}
\titleformat{\section}[block]{\centering\normalfont}{\thesection.}{0.5em}{\uppercase }
\titleformat{\subsection}[runin]{\normalfont}{\thesubsection.}{0.4em plus .1em minus .2em}{\bfseries}[.]
\titleformat{\subsubsection}[runin]{\normalfont}{\thesubsubsection.}{0.4em plus .1em minus .2em}{\it}[.]
\titlespacing*\section{0pt}{18pt plus 4pt minus 2pt}{5pt plus 2pt minus 2pt}
\titlespacing*\subsection{0pt}{10pt plus 2pt minus 1pt}{5pt plus 2pt minus 2pt }
\titlespacing*\subsubsection{0pt}{4pt plus 1pt minus 1pt}{5pt plus 2pt minus 2pt}

\usepackage{chngcntr}
\usepackage{apptools}
\AtAppendix{\counterwithin{lemma}{section}}
\AtAppendix{\counterwithin{assumption}{section}}
\AtAppendix{\counterwithin{corollary}{section}}
\AtAppendix{\counterwithin{theorem}{section}}
\AtAppendix{\counterwithin{proposition}{section}}
\AtAppendix{\counterwithin{table}{section}}
\AtAppendix{\counterwithin{figure}{section}}
\AtAppendix{\counterwithin{remark}{section}}

\makeatletter
\def\mythanks#1{%
	\protected@xdef \@thanks {\@thanks \protect \footnotetext [\the \c@footnote ]{#1}}%
}
\makeatother

\title{Inference in Regression Discontinuity Designs\\ with Clustered Data\mythanks{
		This version: \today.
		We thank Debopam Bhattacharya, Morten Nielsen, Zhuan Pai and numerous seminar and conference participants for helpful comments and suggestions.
		The second author gratefully acknowledges financial support from the European Research Council ERC through grant SH-1852332. 
		Author contact information: 
		Claudia Noack, Department of Economics, University of Bonn. E-Mail: claudia.noack@uni-bonn.de. Website: https://claudianoack.github.io.   
		Tomasz Olma, Department of Statistics, Ludwig Maximilian University of Munich. E-Mail: t.olma@lmu.de. Website: https://tomaszolma.github.io.   
		Christoph Rothe, Department of Economics, University of Mannheim.
		E-Mail: rothe@vwl.uni-mannheim.de. Website: http://www.christophrothe.net. }}

\author{Claudia Noack \and Tomasz Olma \and Christoph Rothe}
\date{}

\begin{document}
	\pagestyle{plain}
	\onehalfspacing
	
	\maketitle 
	
	\begin{abstract}
		Clustered sampling is prevalent in empirical regression discontinuity (RD) designs, but it
		has not received much attention in the theoretical literature. In this paper, we introduce a general model-based framework for such settings and derive high-level conditions under which the standard local linear RD estimator is asymptotically normal. We verify that our high-level assumptions hold across a wide range  of empirical designs, including settings of growing cluster sizes. 
		 We further show that clustered standard errors that are currently used in practice can be either inconsistent or overly conservative in finite samples. To address these issues, 
		 we propose a novel nearest-neighbor-type variance estimator and illustrate its properties in a diverse set of empirical applications.
  \end{abstract}

	\newpage
	\onehalfspacing

	\section{Introduction}
	
Regression discontinuity (RD) designs are widely applied in economics and other social sciences.  
In these settings, treatment is assigned whenever a unit's realization of the running variable crosses a known cutoff; for instance, a candidate is elected only if their vote share exceeds 50\%.
Under continuity conditions on the conditional expectations of the potential outcomes, the average treatment effect at the cutoff is identified by the jump in the conditional expectation of the observed outcome given the running variable at the cutoff. This jump is typically estimated as the difference of the local linear estimates on either side of the cutoff.

	Most theoretical results on estimation and inference in RD designs are derived in settings where the researcher observes a sample of independent and identically distributed (i.i.d.) observations drawn from a large population \citep[e.g.,][]{hahn2001identification, imbens2012optimal, calonico2014robust, armstrong2020simple}. 
	In practice, however, applied researchers often regard the i.i.d.\ assumption as unrealistic and routinely report clustered standard errors. 
	For example, we revisited the survey of recent RD studies published in the journals of the American Economic Association conducted by \citet{noack2021flexible} and found that clustered standard errors were used in around 80\% of the surveyed articles.
	Despite the prevalence of clustered sampling in applied RD studies, formal results for local linear RD estimators -- and more generally local polynomial regression -- under such dependence structures remain limited. In particular, the existing literature offers little guidance on the conditions under which these estimators are asymptotically normal across the range of clustering patterns encountered in practice, or on how to conduct valid inference in such settings. 

	This paper makes two main contributions. First, we provide an asymptotic theory for the local linear RD estimator for clustered data with a large number of independent groups. 
	From a statistical perspective, RD estimators are weighted averages of the outcome variable where the weights depend on the running variable, the kernel, and the bandwidth. When units are clustered, the interaction between local weighting and within-cluster dependence  alters the asymptotic behavior of the local linear RD estimator. We derive high-level conditions under which the local linear RD estimator is asymptotically normally distributed. These conditions are formulated in terms of the weights assigned to units from different clusters and translate into restrictions on cluster sizes within the estimation window.	
	To relate our high-level conditions to empirically relevant designs, we introduce four stylized  asymptotic frameworks motivated by empirical RD applications. These frameworks capture how the asymptotic behavior of local linear RD estimators depends on (i) the effective number of units per cluster within the estimation window, (ii) the dependence structure of the running variable within clusters, and (iii) assumptions on the within-cluster covariance structure of the outcome. 
	Within these frameworks, we find that distinct convergence rates and optimal bandwidth choices are qualitatively different from i.i.d. settings.
	Our results complement the analysis of \citet{hansen2019asymptotic} by considering nonparametric models. In contrast to their analysis, where the estimators are based on the full sample, the localization of the RD estimator leads to non-standard convergence rates.

	Second, we consider estimation of the conditional variance of the local linear RD estimator. We show that the conventional clustered regression residual-based standard error is consistent under the same cluster size conditions that ensure asymptotic normality of the RD estimator.
	However, in settings with independent data, the regression residual-based standard errors are known to exhibit less finite-sample bias than the so-called nearest-neighbors standard errors.
	A naive adaptation of the nearest-neighbors standard error for independent data to clustered settings is invalid, and to our knowledge, no valid nearest-neighbors-type standard error for clustered RD designs currently exists.
	As the second main contribution, we propose a novel clustered nearest-neighbors (CNN) standard error for RD estimators.
	Our proposed method chooses nearest neighbors taking into account the clustering structure and exploiting independence between clusters.	
	We establish consistency of our proposed CNN standard error under our high-level assumptions on the cluster sizes.\footnote{Although our CNN standard error is developed for RD designs in this paper, the general idea behind it extends naturally to other conditional inference problems under misspecification, such as those studied by \citet{abadie2014inference}.}

	We complement our theoretical analysis with empirical applications 
	that assess the finite-sample performance of the proposed standard error relative to existing alternatives.

	\subsection*{Related Literature} 
	Cluster-robust inference is routinely employed in empirical RD designs. Despite this fact, formal results remain limited, even for the standard nonparametric regression using local polynomial estimators. 		
	For RD designs with clustering, \citet{bartalotti2017regression} show asymptotic normality of local polynomial RD estimators under the assumption that all clusters are of the same size and the realizations of the running variable are on the same side of the cutoff within each cluster. 
	Clustered standard errors have been implemented in popular RD packages {\tt rdrobust} and {\tt RDHonest} without much theoretical foundation.
	Our paper contributes to this literature by providing a unified theory for all common RD variants with arbitrary clustering.
	
	\citet{lin2000nonparametric}, \citet{wang2003marginal}, and \citet{bhattacharya2005asymptotic} study general local polynomial estimators under bounded cluster sizes. In these regimes, as the bandwidth converges to zero, the probability of having more than one unit within the estimation window converges to zero for any given cluster, and in consequence, the clustering does not affect the asymptotic distribution.
	\citet{shimizu2024nonparametric} studies nonparametric density and local polynomial estimation under clustered sampling with heterogeneous and potentially unbounded cluster sizes. However, his framework imposes restrictions that rule out several empirically relevant RD settings. In particular, cluster sizes within the estimation window are required to remain uniformly bounded in expectation, the covariates are not allowed to be perfectly dependent within a cluster, and the within-cluster covariance structure of the residuals takes on a very specific form.
	Furthermore, his proposed variance estimator relies on parametric assumptions. By contrast, our framework accommodates weaker conditions on cluster sizes, allows for more general dependence structures in the running variable within clusters, and imposes milder assumptions on the covariance functions, while still delivering asymptotic normality. 
 	Moreover, our proposed nearest-neighbor standard error is fully nonparametric and does not rely on parametric modeling assumptions.
	
	In settings of regular parameters, there is a vast literature about cluster-robust inference.
	\citet{liang1986longitudinal, white2014asymptotic, arellano1987computing}  provide the foundational large-\(G\) theory for cluster-robust covariances for bounded cluster sizes;
	see \citet{cameron2015practitioner, mackinnon2023cluster} for reviews.
	\citet{djogbenou2019asymptotic, hansen2019asymptotic, bugni2025inference, hansen2025jackknife} extend the results to unequal, possibly large clusters.  
	\citet{abadie2023should} discuss clustering adjustments from a design-based perspective. \citet{chiang2025genuinely} point out that in many empirical applications the size of the largest cluster is not negligible relative to the total sample size, and they provide an alternative bootstrap inference procedure.
	
	This paper is the first to propose and formally study a cluster-robust nearest-neighbors-type standard error. In this regard, we extend the work of \cite{abadie2014inference}, who showed consistency of the nearest-neighbors standard errors under i.i.d.\ sampling (in a more general class of misspecified models).
	
\subsection*{Plan of the Paper}
Section~\ref{sec:setting} introduces the model and gives a preview of our main results. Section~\ref{sec:theory_highlevel} states the high-level assumptions and establishes asymptotic normality of the local linear RD estimator. Section~\ref{sec:AFs} studies the asymptotic frameworks. Section~\ref{sec:std_error} introduces our proposed standard error and we show that it is consistent. Section~\ref{sec::numerical_illustrations} contains the empirical applications.
All proofs are collected in the Appendix.

\section{Setting and preview of the results}\label{sec:setting}
\subsection{Clustered Sampling}
	We consider a sharp RD design, in which a unit receives the treatment if and only if their running variable exceeds a known cutoff value, which we normalize to zero.
	The observed data is divided into $G$ clusters, and in each cluster, we observe $n_g$ units. Let $X_{gi}$ and $Y_{gi}$ denote the running variable and the outcome variable of observation $i$ in cluster $g$, respectively. The total sample size is given by $n = \sum_{g \in [G]} n_g$, where $[G] = (1,\ldots,G)$. In our asymptotic analysis, we will treat $G$ and $(n_g)_{g \in [G]}$ as deterministic sequences indexed by the sample size $n$.
	
	Observations in different clusters are independent, but they can be dependent within a cluster.
	The outcome is generated according to the model
	\begin{equation}\label{eq:outcome_equation}
		Y_{gi}= \mu (X_{gi}) + \varepsilon_{gi},
	\end{equation}
	where $\mu(x)= \mathbb E[Y_{gi}|X_{gi}=x]$, $\E[\varepsilon_{gi}|\mc X_g] = 0$, $\mc X_g = (X_{gi})_{i \in I_g}$, and $I_g = (1,\ldots,n_g)$. The error terms $\varepsilon_{gi}$ can be arbitrarily dependent within a cluster. We denote their variance and covariances, conditional on $\mc X_g$ by
	\begin{align*}
		\sigma_{g,i}^2 & = \Var(Y_{gi}|\mathcal X_{g}), \\
		\sigma_{g,ij}  & = \Cov(Y_{gi}, Y_{gj}|\mathcal X_g),
	\end{align*}
	and we denote the covariance matrix of $\mc Y_g = (Y_{gi})_{i \in I_g}$ conditional on $\mc X_g$ by $\Sigma_g$.
	
	Under continuity assumptions on the conditional expectation of the potential outcomes, the jump in the conditional expectation $\mu(x)$ at the cutoff identifies the average treatment effect of units at the cutoff \citep{hahn2001identification}.
	Our parameter of interest is therefore given by
	$$
	\tau = \mu(0^+) - \mu(0^-),
	$$
	where for a generic function $f$, $f(0^+)$ and $f(0^-)$ denote the right and left limits of the function $f$ at zero.	
	
	\begin{remark}
		In the model of observed data in \eqref{eq:outcome_equation}, we assume that the conditional expectation function $\mu$ is the same in each cluster in the sample.
		We note that this model allows for sampling from a population where clusters are heterogeneous in terms of $\mu$ if, in each repeated sample keeping $(n_g)_{g \in [G]}$ fixed, the observed clusters are drawn at random from the superpopulation of clusters of infinite size. Specifically, suppose that we draw a random set of clusters $\mc G = \{\tilde g_1,\ldots, \tilde g_G\}$ and observe a sample of random units from these clusters, $\{(X_{\tilde gi},Y_{\tilde gi})_{i \in [I_{\tilde g}]}\}_{\tilde g \in \mc G }$. If we assume that
		$$
		Y_{\tilde gi}= \mu_{\tilde g}(X_{\tilde g i}) + \eta_{\tilde gi}, \quad \E[\eta_{\tilde gi}|\mc X_{\tilde g}] = 0,
		$$ 
		then this model fits into our framework if we define
		\[
		\mu(x) = \E[\mu_{\tilde g}(x)] \quad\text{and}\quad \varepsilon_{\tilde gi} = \eta_{\tilde gi} + \mu_{\tilde g}(X_{\tilde gi}) - \mu(X_{\tilde gi}),
		\]
		where the expectation is taken w.r.t.\ the distribution of clusters in the population.
	\end{remark}

\subsection{Local Linear RD Estimator}
	In practice, it is common to estimate the RD parameter via local linear RD regressions. This estimator is defined as
	\begin{equation}\label{eq:estimator}
		\wh\tau(h) = e_1^\top\argmin_{\beta\in\IR^4} \sum_{i=1}^n k_h(X_i) (Y_i -  V_i^\top\beta )^2 \equiv \sum_{g \in [G]}\sum_{i \in I_g} w_{gi}(h) Y_{gi},
	\end{equation}
	where $V_i =(T_i, X_i, T_i X_i,1)^\top$, $k_h(v)=k(v/h)/h$ with $k(\cdot)$ a  kernel function  and $h>0$ a bandwidth, and $e_1 = (1,0,0,0)^\top$ is the first unit vector. 
	The weights $w_{gi}(h)$ depend only on the running variable and the bandwidth; the exact expressions for the weights are given in Appendix~\ref{A:sec:weights}.
	In our setting, the variance of the RD estimator conditional on the running variable $\mc X_n = (\mc X_g)_{g \in [G]}$ equals
	\begin{equation}\label{eq:se}
		se^2(h) \equiv \Var(\wh\tau(h) | \mc X_n ) =   \sum_{g \in [G]} \sum_{i,j \in I_g}   w_{gi}(h) w_{gj}(h) \sigma_{g,ij}. 
	\end{equation}
	
\subsection{Preview of Asymptotic Results}
Our first main result establishes the large-sample behavior of the RD estimator under clustered sampling. We derive it under high-level conditions on the weights $w_{gi}(h)$ that translate into restrictions on cluster sizes within the estimation window. Under these general high-level conditions, we show that the local linear RD estimator is asymptotically normal:	
\[
	se(h)^{-1}\big(\widehat{\tau}(h) - \mathbb{E}[\widehat{\tau}(h)\mid \mathcal{X}_n]\big) \xrightarrow{d} \mathcal{N}(0,1).
\]
The convergence rate of the conditional variance $se^2(h)$ depends on the covariance of outcomes within each cluster, $\Sigma_g$, the joint distribution of the realizations of the running variable within each cluster, $\mc X_g$, and the cluster sizes, and it can be slower than the convergence rate of $(nh)^{-1}$ obtained in i.i.d.\ settings.
For example, we show in Section~\ref{sec:AFs} that if the joint distribution of $\mc X_g$ admits a bounded density and some mild regularity conditions hold, then 
\[
	se^2(h) = O_P\left(\frac{1+\lambda_n}{nh}\right),
\]
where
$$
\lambda_n = \frac{h}{n}\sum_{g \in [G]}n_g(n_g-1).
$$
If, in turn, the distribution of  $\mc X_g$ is degenerate, meaning that the realizations of the running variable are equal within each cluster, then
\[
	se^2(h) = O_P\left(\frac{1+\lambda_n/h}{nh}\right).
\] 

The above results provide bounds on the convergence rate of the conditional variance $se^2(h)$. Whether these rates are binding depends on the exact form of the conditional covariance matrix $\Sigma_g$. In Section~\ref{sec:std_error_and_bandwidth_choice}, we give an example where these bounds are achieved, but we note that the convergence rate of $se^2(h)$ can be faster.
For example, in the case of bounded joint density, if the residuals are uncorrelated within each cluster, then $se^2(h) = O_P((nh)^{-1})$, as in the i.i.d.\ case, even if $\lambda_n$ diverges to infinity.

\subsection{Standard Errors}\label{sec:Setting_stderror}	
Reliable and efficient inference requires a variance estimator that is both consistent and well-behaved in finite samples. Natural estimates of $se^2(h)$ are of the form 
\begin{equation}\label{sec:se-hat-overview}
	\wh{se}^2(h) = \sum_{g \in [G]} \sum_{i,j \in I_g} w_{gi}(h) w_{gj}(h)\widehat\sigma_{g,ij}
\end{equation}	
with $\widehat\sigma_{g,ij}$ being some estimate of $\sigma_{g,ij}$. Setting $\wh\sigma_{g,ij}$ to the product of residuals from the local linear RD regression associated with observations $i$ and $j$ in cluster $g$ yields a clustered regression residual-based standard error, analogous to the cluster-robust standard errors proposed for the linear regression \citep{liang1986longitudinal}.
We show that this standard error is valid under the same high-level conditions on the weights as those ensuring asymptotic normality. 

Even though the regression residual-based standard error is consistent under mild assumptions, it has been long recognized in i.i.d.\ settings that this approach may be overly conservative in finite samples, and the nearest-neighbors standard error has been proposed to alleviate this issue \citep{abadie2014inference}. At its core, the nearest-neighbors approach estimates the conditional variance of the outcome variable for any given observation using the outcome variability among its nearest neighbors.

A direct way to adapt the nearest-neighbors variance estimation approach to clustered data is to set $\widehat\sigma_{g,ij}$ in equation~\eqref{sec:se-hat-overview} to
\begin{equation*}
	 \widehat\sigma_{g,ij}^{\text{naive}} = Y^{\Delta,\text{naive}}_{g i} Y^{\Delta,\text{naive}}_{g j}, \quad Y^{\Delta, \text{naive}}_{g i} = \sqrt{ \frac{|\mc N^\text{naive}_{gi}|}{1+|\mc N^\text{naive}_{gi}|} } \left( Y_{g i } - \frac{1}{|\mc N^\text{naive}_{gi}|}\sum_{(g', i') \in \mc N^\text{naive}_{gi} }Y_{g' i'}\right),
\end{equation*}
where $\mc N^\text{naive}_{gi}$ is the set of $J$ units that are closest to unit $i$ in cluster $g$ in terms of the running variable.\footnote{A standard error of that type was proposed by \citet[Supplemental Appendix 7.10]{calonico2019regression} and is implemented in the software package {\tt rdrobust}. We note that this approach also assumes that residuals are uncorrelated across observations on opposite sides of the cutoff.}
However, this procedure does not take into account the clustering structure, and there is no reason to expect that such a variance estimator is consistent in general. First, if the choice of neighbors is associated with cluster membership, then an additional bias term can be present.
Second, the bias-correction factor is devised for variance estimation, but it turns out it is not suitable for covariance estimation. We give two illustrative examples showing these problems in Section~\ref{subsec:failure_naive_CNN}.

To remedy the deficiencies of the naive approach, we propose a different clustered nearest-neighbors standard error where $\widehat\sigma_{g,ij}$ is set to
\begin{equation*}
 \widehat\sigma_{g,ij}^{\sss \text{CNN}} = Y^{\Delta_1}_{g i} Y^{\Delta_2}_{g j}, \quad Y^{\Delta_d}_{g i} \equiv Y_{gi} - \frac{1}{|\mc N^d_{gi} |}\sum_{(g', i') \in \mc N^d_{gi}  }Y_{g'i'} \;\text{ for } d \in \{1,2\}, 
\end{equation*}
where $\mc N^{\,1}_{gi}$ and $\mc N^{\,2}_{gi}$ are carefully chosen sets of neighbors of observation $i$ in cluster $g$. Crucially, we require that 
the observations in $\cup_{i \in I_g}\mc N^{\,1}_{gi}$ and $\cup_{i \in I_g}\mc N^{\,2}_{gi}$ belong to two disjoint sets of clusters not including $g$.
Under this requirement and additional assumptions we show this standard error is consistent, and one can see in simulations that our procedure has favorable finite-sample properties, exhibiting a smaller bias relative to the regression residual-based approach is settings where the curvature of $\mu$ is substantial.

\section{Asymptotic Normality under high-level conditions}\label{sec:theory_highlevel}
In this section, we show asymptotic normality of the local linear RD estimator under high-level assumptions on the weights.
	
	\subsection{High-Level Conditions}
	The first assumption controls the cluster sizes in terms of the weights assigned to units in each cluster. 
	\begin{assumption}\label{ass:HL_cluster_size}
		\begin{enumerate}[label=(\roman*)]
			\item []
			\item $\displaystyle \max_{g \in [G]} \sum_{i,j \in I_g} |w_{gi}(h) w_{gj}(h)| /se^2(h) =o_P(1) $,
			\item $\displaystyle \sum_{g \in [G]} \sum_{i,j \in I_g} |w_{gi}(h) w_{gj}(h)| / se^2(h) =O_P(1)$.
		\end{enumerate}
	\end{assumption}
	
	Assumption~\ref{ass:HL_cluster_size} puts restrictions on the size of the terms $\sum_{i,j \in I_g} |w_{gi}(h) w_{gj}(h)|$. In our asymptotic analysis, it is used to control the contributions of different clusters to the conditional variance $se^2(h)$. This assumption allows the number of non-zero weights within each cluster to grow with the sample size, as long as none of the clusters dominates all others in terms of the sums of the absolute value of cross-products of the weights. This assumption is very general and we provide a range of different asymptotic frameworks in which it is satisfied in Section~\ref{sec:AFs}. However, we note that it is a sufficient condition to obtain asymptotic normality, but it is not necessary. In particular, it does not cover settings where the number of clusters does not increase with the sample size.
	For example, if there is one cluster of weakly-dependent data (e.g., strongly mixing process, as extensively studied in time-series settings), one might still obtain asymptotic normality, but part~(i) of the assumption will generally not hold.
	
	\begin{remark}
		Assumption~\ref{ass:HL_cluster_size} does not require the weights $w_{gi}(h)$ to be derived via a local linear regression. It is compatible with any estimator that is a weighted mean of outcomes where the weights do not depend on the outcome variable, e.g., higher-order local polynomial estimators or the optimized RD estimators \citep{imbens2017optimized, ghosh2025plrd}.
	\end{remark}

	The second assumption controls the moments of the error terms and is standard for results invoking central limit theorems with non-i.i.d.\ data.
	
	\begin{assumption}\label{ass:boundedmoment}
		$\E[|\varepsilon_{gi}|^\delta|\mc X_g]$ is uniformly bounded for some $\delta>2$ for all $g \in [G]$ and $i \in I_g$.
	\end{assumption}
	
	\subsection{Asymptotic Normality}
	Our first main result establishes asymptotic normality of the RD estimator. 	
	To control its bias, we introduce the H{\"o}lder-type class of real functions that are potentially discontinuous at zero, are twice
	differentiable on either 
	side of the threshold, and whose  second derivatives  are uniformly bounded by some constant $M>0$:
	$$\mathcal{F}_H(M)=\{f_1(x)\1{x\geq 0}- f_0(x)\1{x< 0}: \|f_w''\|{\infty} \leq M, w=0,1\}.$$
	
	\begin{theorem}\label{th:normality}
		\begin{enumerate}[label=(\roman*)]
			\item Suppose that Assumptions~\ref{ass:HL_cluster_size} and~\ref{ass:boundedmoment} hold. Then, conditional on $\mc X_n$,
			\[
			\frac{\wh\tau(h) - \E[\wh\tau(h) | \mc X_n]}{se(h)} \xrightarrow{d} \mc N(0,1).
			\]
			\item For any $M \ge 0$,
			\[
				\sup_{\mu \in \mathcal{F}_H(M)}\big|\E[\wh\tau(h) | \mc X_n] - \tau \big| \leq \bar b(h) \equiv -\frac{M}{2} \sum_{g \in [G]} \sum_{i \in I_g} w_{gi}(h) X_{gi}^2 sign(X_{gi}).
			\]
		\end{enumerate}
	\end{theorem}
	Part~(i) of the theorem shows that $\wh\tau(h)$, appropriately recentered and rescaled, is asymptotically normally distributed. Part~(ii) shows that the conditional bias is bounded by a quantity that is the product of the bound on the second derivative of the conditional expectation function and an expression that depends only on the weights and the running variable. The bias bound is the same as in the i.i.d.\ setting \citep{armstrong2020simple, noack2024bias} since the local linear RD estimator is a linear operator and its conditional expectation is not affected by the dependence across outcomes.
	
	Given the standard deviation of the local linear RD estimator and a bound on its bias, we will consider the conditional worst-case mean squared error over data-generating process with $\mu \in \mathcal{F}_H(M)$:
	\[
		\overline{MSE}(h) =  \bar b(h)^2 + se^2(h).
	\]
	We will explicitly derive the limit of $\overline{MSE}(h)$ and characterize the corresponding asymptotically optimal bandwidth under low-level conditions in the next section.

\section{Asymptotic Frameworks}\label{sec:AFs}
Our high-level Assumption~\ref{ass:HL_cluster_size} introduced in Section~\ref{sec:theory_highlevel} accommodates a wide range of empirical clustering structures. To illustrate this flexibility, we formalize four asymptotic frameworks and provide low-level conditions under which the high-level assumption holds. 
Each of the asymptotic frameworks is calibrated to a class of empirical RD applications. They differ in (i) how quickly the number of near-cutoff units can grow inside a cluster, (ii) the dependence structure of the running variable within clusters, and (iii) assumptions on the conditional covariance matrix of the outcome.

\emph{Asymptotic Frameworks~I} and \emph{II} are motivated by empirical applications where there are many clusters that can be potentially large, but the number of units from any cluster within the estimation windows is relatively small. They cover two distinct dependence patterns in the running variable. While Asymptotic Framework I assumes that the joint distribution of the running variable within each cluster admits a bounded joint density, Asymptotic Framework II does not put any restrictions on this distribution at the cost of imposing stronger restrictions on the rates of the cluster sizes.
\emph{Asymptotic Frameworks III} and \emph{IV} relax the restrictions on cluster sizes imposed in Asymptotic Frameworks~I and II, respectively, while instead requiring additional assumptions on the covariance structure of the outcome residuals. We discuss representative empirical examples of each of the asymptotic frameworks in Section~\ref{subsec:empirical_applicaitons}.

\subsection{General Assumptions}
We will first present general assumptions that are maintained in all the four asymptotic frameworks.

\begin{assumption}\label{ass:distX}
	\begin{enumerate*}[label=(\roman*)]
	\item 
	The marginal distribution of $X_{gi}$ is the same for all $g \in [G]$ and $i \in I_g$, and it admits a continuous density $f_X$ that is bounded away from zero around the cutoff.
		\item The kernel function $k$ is a bounded and symmetric density function with bounded support, say $[-1,1]$.
		\item $h \to 0$ and $nh \to \infty$.
	\end{enumerate*}
	
	\end{assumption}

Part~(i) of Assumption~\ref{ass:distX} matches standard conditions for local polynomial estimation with continuously distributed regressors \citep{fan1996local}. We impose continuity (and continuity at the cutoff) only to obtain simple closed-form variance limits; our high-level results can also accommodate discrete, mixed, or cutoff-discontinuous running variables under suitable modifications to the variance calculations.  In the same spirit, we assume the density is continuous at the cutoff to simplify the formulas; formally, an RD analysis can proceed without this requirement, and our main results remain unaffected if the running variable is discontinuous at the threshold.

The kernel and bandwidth requirements in parts (ii) and (iii) are standard in the nonparametric regression literature. We note that the assumption that the bandwidth shrinks to zero is not necessary for our high-level conditions to hold.
We could accommodate a fixed bandwidth; we maintain the usual condition that the bandwidth shrinks to zero solely to obtain closed-form expressions for the leading bias and variance terms. 
In the asymptotic results, we will use the following kernel constants. For $j \in \mathbf{N}$, let $\bar\mu_j = \int_0^1 k(v)v^j dv$. Further, define
$\bar\mu = (\bar{\mu}_2^2 - \bar{\mu}_1 \bar{\mu}_{3})/(\bar{\mu}_2\bar{\mu}_0-\bar{\mu}_1^2)$
and
$\bar\kappa = \int_0^1 \bar k(v)^2 dv$, where $\bar k(v) = k(v)(\bar\mu_2 - \bar\mu_1 v ) / (\bar\mu_2\bar\mu_0 - \bar\mu_1^2)$.

\begin{assumption}\label{ass:epsilon}
	\begin{enumerate*}[label=(\roman*)]
		\item $\E[|\varepsilon_{gi}|^2|\mc X_g]$ is uniformly bounded for all $g \in [G]$ and $i \in I_g$.
		\item The minimal eigenvalue of the conditional covariance matrix of cluster $g$, $\lambda_{min}(\Sigma_g)$, is bounded away from zero uniformly in $n$ and $g \in [G]$.
	\end{enumerate*}
\end{assumption}

Part~(i) is slightly weaker than Assumption~\ref{ass:boundedmoment}.
Part~(ii) is imposed to rule out degenerate conditional dependence structures in the outcome variable, such as situations where two residuals within the same cluster are perfectly negatively correlated. The variance of the RD estimator could collapse to zero in such cases, precluding asymptotic normality.

\subsection{Small Effective Cluster Sizes}
The first two asymptotic frameworks apply to settings where the number of units from any cluster within the estimation window is relatively small.

\subsubsection{Assumptions}
The number of units within the estimation window from any given cluster depends on the total number of units in the cluster as well as the joint distribution of the running variable within the cluster. In Asymptotic Framework~I, we assume that the joint distribution of the realizations of the running variable admits a bounded joint density, while Asymptotic Framework~II leaves the joint distribution unrestricted.

\begin{ass}{AF-I}\label{ass:LLI_clustersize}
	\begin{enumerate*}[label=(\roman*)]
		\item The joint density of $(X_{gi_1},\ldots, X_{gi_k})$ is bounded uniformly in $n \in \mathbb{N}$, $g \in [G]$, and $1 \leq i_1 < \ldots < i_k \leq n_g$ for $k \in \{2,\ldots, n_g\}$.
		\item $\displaystyle\lambda_n = O(1)$.
		\item $\displaystyle\frac{\max_{g \in [G]} n_{g}^2h^2 +  \log^2 G  }{nh} =o(1)$.
	\end{enumerate*}		
\end{ass}	

Assumption~\ref{ass:LLI_clustersize}(i) excludes cases of perfect within-cluster correlation in the running variable where all units share the same realization of the running variable; such extreme dependence is allowed under Assumption~\ref{ass:LLII_clustersize} below at the cost of more restrictive rate conditions. Assumption~\ref{ass:LLI_clustersize} is similar to Assumption~2 of \citet{hansen2019asymptotic}, who study the convergence of the average of clustered observations and other regular, full-sample estimators. Since we consider estimation using the data close to the cutoff, our conditions are formulated in terms of the ``local sample size'' $nh$ and the ``local cluster sizes'' $n_gh$, rather than the full sample size $n$ and the cluster sizes $n_g$ used by \citet{hansen2019asymptotic}. Another conceptual difference is that part~(iii) includes an additional \(\log G\) term. This extra factor is due to the randomness of cluster sizes within the estimation window in our framework, whereas \citet{hansen2019asymptotic} consider the setting of regular parameters and use all observations within each cluster. We note that this asymptotic framework shares some similarities with the framework of \citet{shimizu2024nonparametric} specialized to one continuous covariate. He also assumes that the joint distribution of the realizations of the running variable admits a joint density (albeit only for subsets of four observations) and $\lambda_n=O(1)$. However, he imposes a restrictive assumption that $\max_{g \in [G]} n_g h = O(1)$, whereas our framework allows this quantity to diverge.

\begin{ass}{AF-II}\label{ass:LLII_clustersize}
	\begin{enumerate*}[label=(\roman*)]
		\item $\displaystyle \lambda_n = O(h)$.
		\item $\displaystyle\frac{\max_{g \in [G]} n_{g}^2}{nh}   = o(1)$.			
	\end{enumerate*}
\end{ass}	

Assumption~\ref{ass:LLII_clustersize} imposes no restrictions on the joint density of the running variable, but the rate conditions on the cluster sizes are more stringent than in Assumption~\ref{ass:LLI_clustersize}. In particular, it allows for the realizations of the running variable to be equal within each cluster. It turns out that this extreme case determines the restrictions on the cluster sizes that are necessary to verify Assumption~\ref{ass:HL_cluster_size}.

To illustrate the restrictions imposed in Assumptions~\ref{ass:LLI_clustersize} and~\ref{ass:LLII_clustersize}, we consider two examples that differ in the degree of allowed heterogeneity in cluster sizes. We revisit these examples in the next subsection to facilitate comparisons between asymptotic frameworks.

\begin{exampleA}
	Suppose that all clusters are of the same size, i.e., $n_g=n/G$ for all $g \in [G]$, then $\lambda_n = O(n_1h)$. The rate conditions in Assumption~\ref{ass:LLI_clustersize} reduce to 
	$n_1h = O(1)$  and $\log^2 G/(nh) = o(1)$ and in 	Assumption~\ref{ass:LLII_clustersize}  to $n_1 = O(1)$ and $Gh \to \infty$.
\end{exampleA}

We note that in Example~1, the expected number of units within the estimation window remains bounded in any given cluster. The next example shows that both Asymptotic Frameworks~I and~II allow the maximal expected number of observations within the estimation window to diverge for some clusters.

\begin{exampleB}
	Let $a\geq b$ for $a,b \in (0,1)$.
	Suppose that we observe	$G_1 = n-\lfloor n^{a}\rfloor$ clusters of size $1$ and $G_2=\lfloor n^{b}\rfloor$ clusters of size $\lfloor n^{a-b}\rfloor$. Further, for illustration, assume that $h=n^{-b}$. 
	Then $\lambda_n = O(n^{-b} + n^{2(a-b)-1} )$. The rate conditions of Assumption~\ref{ass:LLI_clustersize} hold if
	$	 a \leq b+\frac{1}{2}$. In addition, if $a>2b$, then each large cluster contributes a diverging number of units local to the cutoff, since
	$\max_{g \in [G]} n_g h \asymp n^{a-2b}\to\infty$. The rate conditions of Assumption~\ref{ass:LLII_clustersize} holds provided
	$a < \frac{1}{2}(1+b)$.
\end{exampleB}

\subsubsection{Theoretical results}
The following proposition presents our key theoretical results for Asymptotic Frameworks I and~II.

\begin{proposition}\label{prop:AF_I-II}
	Suppose that Assumptions~\ref{ass:distX},~\ref{ass:epsilon}, and either Assumption~\ref{ass:LLI_clustersize} or~\ref{ass:LLII_clustersize}  holds. Then Assumption~\ref{ass:HL_cluster_size} is satisfied, $se^2(h)\asymp_p (nh)^{-1}$, and $\bar b(h) = - M \bar\mu h^2(1 + o_P(1))$.
\end{proposition}

Proposition~\ref{prop:AF_I-II} verifies our high-level condition on the weights, shows that the conditional variance of the RD estimator is of order $(nh)^{-1}$, and it characterizes the leading term of the conditional worst-case bias $\bar b(h)$. It follows that worst-case asymptotic mean squared error is of order $h^4 + (nh)^{-1}$, which is minimized for bandwidths of order $n^{-1/5}$. With such a bandwidth, the estimator converges at the rate $n^{-2/5}$, given that our assumptions hold for this bandwidth choice. 
We note that the order of the variance is the same as in the i.i.d.\ case, but its exact form is in general affected by clustering; we derive its limit under additional assumptions in Section~\ref{sec:std_error_and_bandwidth_choice}.


\subsection{Large Effective Cluster Sizes}
While Asymptotic Frameworks~I and~II cover many relevant clustering patterns, the conditions on the growth rates of the cluster sizes might be restrictive in some applications. 
In this section, we show these rate conditions can be significantly relaxed under direct assumptions on the standard error. We then argue in Section~\ref{sec:std_error_and_bandwidth_choice} that these assumptions are reasonable in many settings.

\subsubsection{Assumptions}
Asymptotic Framework~III considers cases where the joint distribution of the realizations of the running variable is continuous, as in Asymptotic Framework~I, but it relaxes the rate conditions imposed on the cluster sizes. This asymptotic framework is motivated by settings where the clusters have a large number of units even in a shrinking neighborhood of the cutoff.

\begin{ass}{AF-III}\label{ass:LL3_clustersize}
	\begin{enumerate}[label=(\roman*)]
		\item The joint density of $(X_{gi_1},\ldots, X_{gi_k})$ is bounded uniformly in $n \in \mathbb{N}$, $g \in [G]$, and $1 \leq i_1 < \ldots < i_k \leq n_g$ for $k \in \{2,\ldots, n_g\}$.
		\item $\displaystyle \frac{\lambda_n}{nh} +  \frac{\max_{g \in [G]} (n_{g}h)^2 + \log^2 G }{nh (1+\lambda_n)}   = o(1) $.
		
		\item $se^2(h) \asymp_p  (1+\lambda_n)/(nh)$.
	\end{enumerate}
\end{ass}	
Part~(i) of Assumption~\ref{ass:LL3_clustersize} coincides with part~(i) of Assumption~\ref{ass:LLI_clustersize}. However, the
rate conditions on cluster sizes in part (ii) are considerably weaker than those imposed
in parts~(ii) and~(iii) of Assumption~\ref{ass:LLI_clustersize}. In particular, $\lambda_n$ is allowed to diverge to infinity within this framework.

Asymptotic Framework IV applies to settings where each cluster may contain many units whose realizations of the running variable might be highly correlated. As a result, even within a shrinking neighborhood of the cutoff, some clusters can contribute a large number of units concentrated in a narrow region of the support of the running variable. Such settings occur naturally, for example, if the outcomes are measured at the individual level, whereas the running variable is assigned at the cluster level.

\begin{ass}{AF-IV}\label{ass:LL4_clustersize}
	\begin{enumerate*}[label=(\roman*), wide]		
		\item $\displaystyle \frac{\lambda_n}{n h} + \frac{\max_{g \in [G]} n_g^2}{\sum_{g \in [G]} n_g^2} = o(h)$. 
		\item $se^2(h) \asymp_p (1+\lambda_n/h)/(nh)$.
	\end{enumerate*}
\end{ass}

The rate conditions on cluster sizes are significantly weaker than those in Assumption~\ref{ass:LLII_clustersize} as we combine these conditions with an additional assumption on the standard error. 
We note that the rate conditions in Assumption~\ref{ass:LL4_clustersize} are stronger than those in Assumption~\ref{ass:LL3_clustersize}; this is needed because we do not impose any assumptions on the joint distribution of the running variables.
We next illustrate the implications of these rate conditions in our examples studied above.

\begin{exampleA}[Equal Cluster Sizes, cont'd]
	Assumption~\ref{ass:LL3_clustersize} requires that $G \to \infty$ and $n^2h^2/G \to \infty$.
	In this setting, $se^2(h) = O\left(1/(nh) +  1/G \right) $. Part~(i) of Assumption~\ref{ass:LL4_clustersize} reduces to the requirement that $Gh \to \infty$. We emphasize that there is no separate restriction on the cluster size. Part~(ii) implies that 	$	se^2(h) = O(1/(Gh))$.
\end{exampleA}

\begin{exampleB}[Heterogeneous Cluster Sizes, cont'd]
	The rate conditions of 	Assumption~\ref{ass:LL3_clustersize} do not impose any additional restrictions. Moreover, if $a>2b$, at least one large cluster is locally influential, in the sense that $\max_{g \in [G]} n_g h\asymp n^{a-2b}\to\infty$.	The rate conditions of 	Assumption~\ref{ass:LL4_clustersize}  impose  the same restrictions as in Asymptotic Framework II. It has to hold that $a<\frac{1}{2}(1+b)$.
\end{exampleB}

\subsubsection{Theoretical results}
The following proposition presents our key theoretical results for Asymptotic Frameworks III and~IV.
\begin{proposition}\label{prop:AF_III-IV}
	Suppose that Assumptions~\ref{ass:distX} and~\ref{ass:epsilon} hold, and either Assumption~\ref{ass:LL3_clustersize} or~\ref{ass:LL4_clustersize}  holds. Then Assumption~\ref{ass:HL_cluster_size} is satisfied and $\bar b(h) = - M \bar\mu h^2(1 + o_P(1))$.
\end{proposition}

Proposition~\ref{prop:AF_III-IV} verifies our high-level Assumption~\ref{ass:HL_cluster_size} and characterizes the leading term of the conditional worst-case bias $\bar b(h)$. We note that the rate of the optimal bandwidth and the resulting convergence rate of the estimator are different than in the i.i.d.\ case or the case of small effective cluster sizes discussed in the previous section. Specifically, in Framework III, we have that
\[
	\wh \tau(h) - \tau = O_P\left(h^2 + \frac{1}{\sqrt{nh}} + \sqrt{ \frac{\sum_{g \in [G]}n_g^2}{n^2} } \right).
\]
In contrast to Asymptotic Framework I, the third term in the remainder on the right-hands side can dominate the other terms. Since the bandwidth $h$ appears only in the first two terms, the convergence rate is optimized if $h \sim n^{-1/5}$, in which case we obtain $\wh \tau(h) - \tau = O_P(n^{-2/5} +  \big(\sum_{g \in [G]}n_g^2 / n^2 \big)^{1/2} )$. We note that if $n^{-2/5} = o\big(  \big( \sum_{g \in [G]}n_g^2 / n^2 \big)^{1/2} \big)$ and  $h\sim n^{-1/5}$, then the variance dominates the bias under the optimal bandwidth choice, such that 
\[
	\frac{\wh\tau(h) - \tau}{se(h)} \xrightarrow{d} \mc N(0,1),
\]
assuming that Assumption~\ref{ass:LL3_clustersize} holds for this bandwidth choice.

In Framework IV, in turn, we have that
\[
\wh \tau(h) - \tau = O_P\left(h^2 + \frac{1}{\sqrt{nh}} \left(1  + \sqrt{ \frac{\sum_{g \in [G]}n_g^2}{n} } \right) \right).
\]
The fastest possible convergence rate is achieved for $ h \sim  \big(\big(1 + \sum_{g \in [G]}n_g^2/n\big)/n \big)^{1/5}$, yielding $\wh \tau(h) - \tau = O_P\big( n^{-2/5} + \big(\sum_{g \in [G]}n_g^2/n^2\big)^{2/5} \big) $, given that this bandwidth choice satisfies Assumption~\ref{ass:LL4_clustersize}.

\subsection{Standard Error and Bandwidth Choice in a Special Case}\label{sec:std_error_and_bandwidth_choice}
To illustrate the behavior of the local linear RD estimator with clustered data, we will further  consider a simplified setup, where exact characterization of the limit of the variance is possible. The following assumption formalizes the setup.
\begin{assumption}\label{ass:limit_variance}
	For all $g\in [G]$ and $i \in I_g$, $\sigma^2_{i,g}=\sigma^2(X_{gi})$ and $\sigma_{ij,g} = \sigma(X_{gi},X_{gj})$ for some functions $\sigma^2(\cdot)$ and $\sigma(\cdot,\cdot)$ that are $L$-Lipschitz continuous away from the cutoff.
\end{assumption}
Assumption~\ref{ass:limit_variance} imposes a common covariance between any two units and across clusters, and it imposes continuity of the conditional variance and covariance functions.\footnote{This type of assumption can be rationalized by models studied in the functional data literature \citep[see, e.g.,][]{zhang2007statistical}. \cite{shimizu2024nonparametric} also imposes such an assumption.} 

\subsubsection{Continuous Joint Distribution of the Running Variable} We first study the standard error under Asymptotic Frameworks I and III.

\begin{lemma}\label{lemma:se_1}
	Suppose that Assumptions~\ref{ass:distX},~\ref{ass:epsilon}(i), and~\ref{ass:LL3_clustersize}(i)-(ii) hold.
	Then
	\begin{enumerate}[label=(\roman*)]
		\item $\displaystyle se^2(h) = O_P\left(\frac{1 + \lambda_n}{nh} \right)$.
		\item If in addition Assumption~\ref{ass:limit_variance} holds and $(X_{gi},X_{gj})$, $i\neq j$, $i,j \in I_g$, are identically distributed with continuous
		joint density $f(x_1,x_2)$,	then
		\begin{align*}
			se^2(h) & =  \frac{ 1}{nh} \left( \bar \kappa \sum_{\star \in \{+,-\}} \frac{\sigma^2(0^\star)}{f_X(0)} 
			+ \lambda_n  \sum_{\star, \diamond \in \{+, -\}} \onepm{\star=\diamond} \sigma(0^\star,0^\diamond)  \frac{ f(0, 0)}{f_X(0)^2} + o_P\left(1 + \lambda_n\right) \right),
		\end{align*}	
		where $\onepm{\star=\diamond}=1$ if $\star=\diamond$ and $\onepm{\star=\diamond} = -1$ otherwise.
	\end{enumerate}
\end{lemma}	

Lemma~\ref{lemma:se_1} imposes only the weak rate conditions on the cluster sizes of Asymptotic Framework~III, which, in particular, cover Asymptotic Framework~I. First, the lemma provides an upper bound on the convergence rate of the conditional variance $se^2(h)$, and then it derives its exact limit in a special case.
Part~(ii) demonstrates that the conditional variance is asymptotically equal to the sum of two terms. The first one is driven by the variances of individual units and is the same as in the i.i.d.\ case. The second component is due to the covariance between outcomes within a cluster; we note this part depends neither on the bandwidth nor the choice of the kernel function. 
If  $\lambda_n$ converges to a positive constant, then the two terms are of the same order.
If $\lambda_n$ converges to zero, then the effect of clustering on variance becomes asymptotically negligible; a similar result was obtained by \citet{shimizu2024nonparametric} under stronger rate restrictions on the cluster sizes. If $\lambda_n$ diverges to infinity and  $\sum_{\star, \diamond \in \{+,-\}} \onepm{\star = \diamond} \sigma(0^\star, 0^\diamond) \neq 0$, the covariance part dominates. 

Under the assumptions of part~(ii) of Lemma~\ref{lemma:se_1}, the worst-case mean squared error of $\wh\tau(h)$ satisfies
\[
\overline{MSE}(h) = \left(M^2\bar\mu^2h^4 + \frac{V_1}{nh} + \frac{\sum_{g \in [G]} n_g^2}{n^2} \sum_{\star, \diamond \in \{+, -\}} \onepm{\star=\diamond} \sigma(0^\star,0^\diamond)  \frac{ f(0, 0)}{f_X(0)^2} \right)(1+o_P(1)) ,
\]
where $V_1 =   \frac{\bar \kappa}{f_X(0)} \sum_{\star \in \{+,-\}} \sigma^2(0^\star)$.
The bandwidth minimizing the leading term is given by
$$
h^*_1 = \left(\frac{V_1}{4M^2\bar\mu^2}\right)^{1/5} n^{-1/5},
$$
given that this bandwidth choice satisfied the assumptions of part~(ii) of Lemma~\ref{lemma:se_2}.
The optimal bandwidth $h^*_1$ is the same as the AMSE-optimal bandwidth in the i.i.d.\ case. We emphasize that this result relies on Assumption~\ref{ass:limit_variance}; under more general covariance structures, it need not hold.

\subsubsection{Degenerate Joint Distribution of the Running Variable} We now study the standard error under Asymptotic Frameworks II and IV.
To derive a closed-form expression for the limit of the conditional variance in this setting, we need to impose an additional assumption on the joint distribution of the running variable. For illustration, we consider the setting where all the realizations of the running variable are the same within cluster.
\begin{lemma}\label{lemma:se_2}
	Suppose that Assumptions~\ref{ass:distX},~\ref{ass:epsilon}, and~\ref{ass:LL4_clustersize}(i) hold. Then
	\begin{enumerate}[label=(\roman*)]
	\item $\displaystyle se^2(h) = O_P\left(\frac{1 + \lambda_n/h }{nh} \right) $
	
	\item If in addition Assumption~\ref{ass:limit_variance} holds and the realizations of the running variable are equal within each cluster. Then
	\[
	se^2(h) = 	\frac{1}{nh} \frac{\bar\kappa}{f_X(0)} \sum_{\star \in \{+,-\}}  \bigg(  \sigma^2(0^\star)  +   \frac{\lambda_n}{h}  \sigma(0^\star,0^\star) + o_P\left(1 + \frac{\lambda_n}{h} \right)   \bigg).
	\]
	\end{enumerate}
\end{lemma}	
Lemma~\ref{lemma:se_2} imposes only the weak rate conditions on the cluster sizes of Asymptotic Framework~IV, which, in particular, cover Asymptotic Framework~II.
First, the lemma provides an upper bound on the convergence rate of the conditional variance $se^2(h)$, and then it derives its exact limit in a special case.

Under the assumptions of Lemma~\ref{lemma:se_2}, the worst-case mean squared error of $\wh\tau(h)$ satisfies
\[
	\overline{MSE}(h) = \left(M^2\bar\mu^2h^4 + \frac{V_{2}}{nh}\right)(1+o_P(1)),
\]
where $ V_2 =   \frac{\bar\kappa}{f_X(0)} \sum_{\star \in \{+,-\}}  \big(  \sigma^2(0^\star)  + \sigma(0^\star,0^\star) \sum_{g \in [G]}n_g(n_g-1)/n  \big)$.
In this setting, the bandwidth minimizing the leading term of the AMSE-optimal bandwidth is given by
$$
h^*_2 = \left(\frac{V_2}{4M^2\bar\mu^2}\right)^{1/5} n^{-1/5},
$$
assuming the assumptions of Lemma~\ref{lemma:se_2} hold for this bandwidth choice.

\section{Clustered Standard Error}\label{sec:std_error}
In this section, we study variance estimation based on the nearest-neighbors and regression residual-based approaches.

\subsection{Clustered Nearest-Neighbors Standard Error}\label{sec:CNN}
We first show why the naive adaptation of the nearest-neighbors approach devised for i.i.d.\ settings is in general not valid with clustered data. We then introduce our proposed clustered nearest-neighbors (CNN) standard error and show its consistency.

\subsubsection{Failure of the Naive Clustered Nearest-Neighbors Standard Error}\label{subsec:failure_naive_CNN}
To highlight the main problems with the naive clustered nearest-neighbors standard error described in Section~\ref{sec:Setting_stderror}, we consider a simple setup with $\mu(X) = 0$, $\Var(Y_{gi}|\mc X_g) = \sigma^2$, and $J=1$ nearest neighbor. We discuss two examples that differ in the assumptions on the joint distribution of the realizations of the running variable.

First, suppose that each cluster consists of only two observations and $X_{g1}=X_{g2}$ for all $g \in [G]$, such that $\mc N^{\text{naive}}_{g1} = \{(g,2)\}$ and $\mc N^{\text{naive}}_{g2} = \{(g,1)\}$. Then
	\begin{align*}
		\wh{se}_{\text{naive} }^2(h) &  =  \frac{1}{2}  \sum_{g \in [G]}     w_{g1}^2(h)  \sum_{i,j \in I_g } Y^{\Delta, \text{naive}}_{g i} Y^{\Delta, \text{naive}}_{g j} \\ 
		& = \frac{1}{2} \sum_{g \in [G]}  w_{g1}^2(h) \underbrace{ \sum_{l=1}^4 (-1)^{l}  (\varepsilon_{g1} - \varepsilon_{g2})^2 }_{=0} =0.
	\end{align*}
Clearly, this standard error cannot be consistent. While this is a very specific example, the same type of problem arises more generally whenever the realizations of the running variable are highly concentrated within clusters.

Second, suppose that the joint distribution of the running variable within each cluster admits a bounded joint density.
Consider two observations $i,j \in I_g$, $i \neq j$, and let $(g_1, i')$ and $(g_2, j')$ be their respective nearest neighbors. Assume further that the clusters $g$, $g_1$, and $g_2$ are pairwise different, which occurs with high probability in this setup if all clusters are relatively small. Then, while it is easy to see that the conditional variance estimate is correctly centered, $\mathbb E\big[ \big(Y^{\Delta, \text{naive}}_{gi}\big)^2 |\mathcal X_n\big] = \sigma^2$, the conditional covariance estimates are biased:
		\begin{align*}
			\mathbb E\Big[ Y^{\Delta, \text{naive}}_{gi} Y^{\Delta, \text{naive}}_{gj} |\mathcal X_n\Big] & = \frac{1}{2} \mathbb E[ \varepsilon_{gi}\varepsilon_{gj} - \varepsilon_{gi}\varepsilon_{g_2j'}  - \varepsilon_{gj}\varepsilon_{g_1i'}	+ \varepsilon_{g_1i'}\varepsilon_{g_2j'} | \mathcal X_n] 
			= \frac{1}{2} \mathbb E[ \varepsilon_{gi}\varepsilon_{gj} | \mathcal X_n].
		\end{align*}
It follows that $\wh{se}_{\text{naive}}^2(h)$ is not correctly centered in general.\footnote{We note that the problem arising in the naive conditional covariance estimation can be easily resolved by leaving out the correction factor when $i \ne j$. However, a complete proof of approximate unbiasedness of the standard error would still require controlling the probability of the respective clusters being distinct across all observations, limiting the applicability of this method to relatively small clusters with a bounded joint density of the realizations of the running variable.}

\subsubsection{Our Proposed Clustered Nearest-Neighbors Standard Error}\label{subsec:SE_descitpion}
It is evident from the preceding discussion that selecting neighbors from distinct clusters induces certain independence restrictions that are instrumental for establishing conditional unbiasedness of the standard error. We leverage this insight to construct our proposed standard error, explicitly enforcing the desired independence structure.

For every cluster $g\in [G]$, we define two sets of its ``companion clusters'', $\mc R_g^1 \subset [G]$ and $\mc R_g^2 \subset [G]$. Next, for every $i \in I_g$ and $d \in \{1,2\}$, we define $\mc N_{gi}^{d}$ as the set of at least $J$ nearest neighbors of unit $i$ in cluster $g$ in terms of the running variable that are on the same side of the cutoff as $X_{gi}$ and belong to a cluster in the set $\mc R_g^d$. 
Our proposed clustered nearest-neighbors (CNN) standard error is then defined as:
\begin{equation}
	\wh{se}_{\sss CNN}^2(h) = \sum_{g \in [G]} \sum_{i,j \in I_g } w_{gi}(h) w_{gj}(h) Y^{\Delta_1}_{g i} Y^{\Delta_2}_{g j}, \quad Y^{\Delta_d}_{g i} \equiv Y_{gi} - \frac{1}{|\mc N_{gi}^{d}|}\sum_{(g', i') \in \mc N_{gi}^{d} }Y_{g'i'}.
\end{equation}
To show consistency of this standard error, we require the sets of nearest neighbors to satisfy two properties. First we require that the neighbors in the two sets $\mc N_{gi}^{\,1}$ and $\mc N_{gi}^{\,2}$ are independent of the units in cluster $g$ and of each other. This is achieved by selecting companion clusters that do not include cluster $g$, $g \notin \mc R^1_g \cup \mc R^2_g$, and are disjoint, $\mc R^1_g \cap \mc R^2_g = \emptyset$. These properties ensure that our standard error is asymptotically conditionally unbiased.

Second, we need to control the dependence between $\sum_{i,j \in I_g} w_{gi}(h) w_{gj}(h) Y^{\Delta_1}_{g i} Y^{\Delta_2}_{g j}$ across different clusters.
We achieve that by imposing that each cluster can be a  companion cluster for at most $R$ clusters for some fixed number $R$, i.e., for all $g \in [G]$, we require
\begin{equation}
\#\bigl\{\tilde g \in [G] : g \in \mc R^1_{\tilde g} \cup \mc R^2_{\tilde g}\bigr\} \leq R.\label{equ:max_R}	
\end{equation}
These general requirements on the choice of companion clusters can be satisfied by many different selection methods. We provide one concrete algorithm in Appendix~\ref{sec:selection_algorithms}.

\begin{remark}
	We note that $\wh{se}_{\sss CNN}^2(h)$ does not reduce to the nearest-neighbors standard error that is typically used in i.i.d.\ settings. The standard nearest-neighbors standard errors replaces the unknown conditional variances with the squares of distances to nearest neighbors and it involves a bias-correction factor. Since we effectively take the product of two different residuals, it turns out that this bias-correction factor is not needed. 
\end{remark}

\subsubsection{Consistency of the CNN Standard Error}
As is standard for nearest-neighbors-type estimators, our method relies on the assumption that the nearest neighbors selected in the construction of $\wh{se}_{\sss CNN}^2(h)$ are uniformly close to the respective units.

\begin{assumption}\label{ass:HL_closest_best_neighbor}	
	$
	\displaystyle D(h) \equiv \max_{g \in [G]}  \max_{ \substack{i \in I_g \\ w_{gi}(h) \neq 0}}   \max_{(\tilde g, j) \in \mathcal N^{\,1}_{gi} \cup \mathcal N_{gi}^{\,2} } |X_{gi} - X_{\tilde g j}| = o_P(1).
	$
\end{assumption}
This assumption is automatically satisfied whenever the nearest neighbors are selected within the estimation window, as is the case for the algorithm given in Appendix~\ref{sec:selection_algorithms}, and the bandwidth converges to zero. Since all matched units then lie within distance $h$ of each other, we have $D(h) = O(h)$ by construction.
In general, we expect $D(h)$ to converge to zero at a much faster rate. For illustration, consider a setting where the marginal density of $X_{gi}$ is bounded and bounded away from zero in a neighborhood of the cutoff. If the realizations of the running variable are constant within each cluster and $Gh \to \infty$, then we expect that $D(h) = O_P(\log(Gh)/G)$. As a different example, consider a setting where the joint density of $\mc X_g$ is continuous for all $g \in [G]$ and $n_{\min}h\equiv\min_{g \in [G]} n_gh \to \infty$, then we expect that $D(h) = O_P(\log(n_{min}h)/n_{min})$.

Our second main result states that $\wh{se}_{\sss CNN}^2(h)$ is consistent for $se^2(h)$.

\begin{theorem}\label{theorem::consistency_of_se}
	Suppose that Assumptions~\ref{ass:HL_cluster_size} and~\ref{ass:HL_closest_best_neighbor} hold and Assumption~\ref{ass:boundedmoment} holds with $\delta=4$. Then
	\[
	\frac{\wh{se}_{\sss CNN}^2(h)}{se^2(h)} = 1 + o_{P,\mc F}(1),
	\]
	where the term $o_{P,\mc F}(1)$ converges to zero uniformly over any class of DGPs where $\mu$ is L-Lipschitz continuous away from the cutoff for some constant $L$.
\end{theorem}

\begin{remark}
	In contrast to the standard nearest neighbor variance estimator, that is typically applied in settings of independent samples, we do not need to impose continuity assumptions of the conditional variance function.
\end{remark}

\begin{remark}
	It is common practice in RD designs, to impose that $\mu$ has a bounded second derivative. If one wants to impose this assumption instead of assuming that $\mu$ is L-Lipschitz continuous away from the cutoff for some constant $L$, we can easily modify the standard error following the suggestions of  \citet{noack2024bias}.  
\end{remark}

\subsection{Clustered Regression Residual-Based Standard Error}\label{sec::residual_based_c}
In this subsection, we show consistency of the clustered regression residual-based (CRR) standard error for the local linear RD estimator.
For $\star \in \{+,-\}$, define $b_0^\star = \mu(0^\star)$ and $b_1^\star = \mu'(0^\star)$, and let $\hat b_0^\star$ and $\hat b_1^\star$ denote the intercept and slope coefficient on the respective side of the cutoff in the local linear RD regression in equation~\eqref{eq:estimator}.
The CRR standard error is defined as
	\[
	\wh{se}_{\sss CRR}^2 = \sum_{g \in [G]} \sum_{i,j \in I_g} w_{gi}(h)w_{gj}(h) \hat\varepsilon_{gi} \hat\varepsilon_{gj},\quad \hat\varepsilon_{gi} = Y_{gi} - \wh\mu(X_{gi}),
	\]
	where $\wh\mu(x) =(\hat b_0^- + \hat b_1^- x) \1{x < 0}  + (\hat b_0^+ + \hat b_1^+ x) \1{0 \leq x} $. 
	
	\begin{theorem}\label{theorem:RR_SE}
		Suppose that Assumption~\ref{ass:HL_cluster_size} holds and  Assumption~\ref{ass:boundedmoment} holds with $\delta=4$. Further, assume that $\mu \in \mc F_H(M)$, $\hat b_0^\star - b_0^\star = o_P(1)$,
		$h(\hat b_1^\star - b_1^\star) = o_P(1)$ for $\star \in \{+,-\}$, the weights $w_{gi}(h)$ are zero whenever $|X_{gi}|>h$, and $h \to 0$.
		Then
		\[
		\frac{\wh{se}_{\sss CRR}^2(h)}{se^2(h)} = 1 + o_P(1).			
		\]
	\end{theorem}

Theorem~\ref{theorem:RR_SE} imposes the same high-level assumption on the weights $w_{gi}(h)$ as we used to show consistency of the CNN standard error. The consistency requirements for $\hat b_0^\star$ and $\hat b_1^\star$ are very mild and they are satisfied in all the considered asymptotic frameworks given our smoothness assumption. The main difference in assumptions relative to the result for the CNN standard error is that Theorem~\ref{theorem:RR_SE} requires the bandwidth to converge to zero, while the assumptions of Theorem~\ref{theorem::consistency_of_se} may hold even for an asymptotically fixed bandwidth.

\section{Numerical Illustrations}\label{sec::numerical_illustrations}
In this section, we apply our clustered nearest-neighbors (CNN) standard error in four empirical applications, 
and we compare it to four alternative standard errors.
The first two, the classical nearest-neighbors (NN) and Eicker-Huber-White (EHW) standard errors, do not account for clustering. The third approach is the naive clustered nearest-neighbors (Naive CNN) approach described in Section~\ref{subsec:failure_naive_CNN}, and the fourth is the clustered regression
residual-based (CRR) standard error described in Section~\ref{sec::residual_based_c}.

In order to connect the asymptotic frameworks introduced in Section~\ref{sec:AFs} to observable features of the data, we first propose a simple diagnostic rule of thumb for assessing whether clusters are sufficiently small and reasonably balanced for Asymptotic Framework~I or~II to provide accurate approximations to the underlying data-generating process.
We then revisit four recent RD applications, each exemplifying one of the asymptotic frameworks.\footnote{For each application, we use the data provided in the respective replication package and consider one of the main RD specifications reported in the paper. For simplicity, we ignored any additional covariates that were included to improve estimation precision. We fix the bandwidth at the value used by the authors of the original study. To simplify the analysis, when the original paper used two different bandwidths on each side of the cutoff, we chose the bigger one.} 

\subsection{Rule of Thumb for the High-Level Conditions}\label{subsec::rot}
In this section, we provide a practical rule of thumb to assess whether our high-level Assumption~\ref{ass:HL_cluster_size} on the weights can be plausibly satisfied in a given empirical application. When our conditions hold, the data-generating process may be well approximated by Asymptotic Frameworks I or II under mild assumptions on the covariance of the residuals. If, in turn, these conditions are violated, one may need to justify stronger assumptions on the dependence structure of the residuals to ensure that the standard error converges at a suitable rate, as discussed in our Asymptotic Frameworks III and IV.

To describe our proposed criterion, for $g \in [G]$, define
\[
w_{g,\textnormal{ratio}}(h) = \frac{\sum_{i,j \in I_g} \lvert w_{gi}(h) w_{gj}(h) \rvert}
{\sum_{g \in [G]} \sum_{i \in I_g} w_{gi}(h)^2},
\]
and let
\[
w_{\max}(h)
\equiv
\max_{g \in [G]} w_{g,\textnormal{ratio}}(h),
\qquad
w_{\textnormal{sum}}(h)
\equiv
\sum_{g \in [G]} w_{g,\textnormal{ratio}}(h).
\]
If the conditional covariance matrix of the residuals has eigenvalues bounded away from zero, the conditions
$w_{\max}(h) = o_{P}(1)$ and $w_{\textnormal{sum}}(h) = O_{P}(1)$
are sufficient to ensure that our high-level Assumption~\ref{ass:HL_cluster_size} is satisfied. To operationalize these asymptotic conditions, in finite samples, one needs to choose some threshold values $\eta_{\text{max}}$ and  $\eta_{\text{sum}}$, and check whether $w_{\max}(h) \le \eta_{\max}$ and $w_{\textnormal{sum}}(h) \le \eta_{\textnormal{sum}}$.
We consider $\eta_{\max} = 0.1$ and $\eta_{\textnormal{sum}} = 10$ to be reasonable benchmark values in practice.\footnote{To give a heuristic interpretation of these threshold values, 
we note that, asymptotically, $w_{\max}(h) \approx  \max_{g \in [G]} n_{g,h}^2 / n_h$ and $w_{\text{sum}}(h) \approx \sum_{g \in [G]} n_{g,h}^2/n_h$,	where \( n_{g,h} \) and $n_h$ denote the number of units from cluster \( g \) and the total number of units within the estimation window, respectively. Now suppose that within the estimation window, there are 100 clusters with 10 observations each, a setting in which asymptotic normality can plausibly be a good approximation. Then $w_{\max}(h) \approx 0.1$ and $w_{\text{sum}}(h) \approx 10$.}

\subsection{Empirical Applications}\label{subsec:empirical_applicaitons}
In this section, we revisit four empirical applications that motivated the asymptotic frameworks introduced in Section~\ref{sec:AFs}.
Figure~\ref{figure::empirical_examples_running_variable_cluster_sizes} illustrates the number of clusters, the distribution of cluster sizes, and the distribution of the running variable within a cluster. Each point represents an individual observation. For discrete outcome variables, values are jittered to improve visual clarity and mitigate overplotting. In selected clusters, all observations are displayed in a common color and use the same marker symbol to emphasize cluster membership. The dashed lines mark the bandwidths used.

	\begin{figure}

	\includegraphics[width=0.96\textwidth]{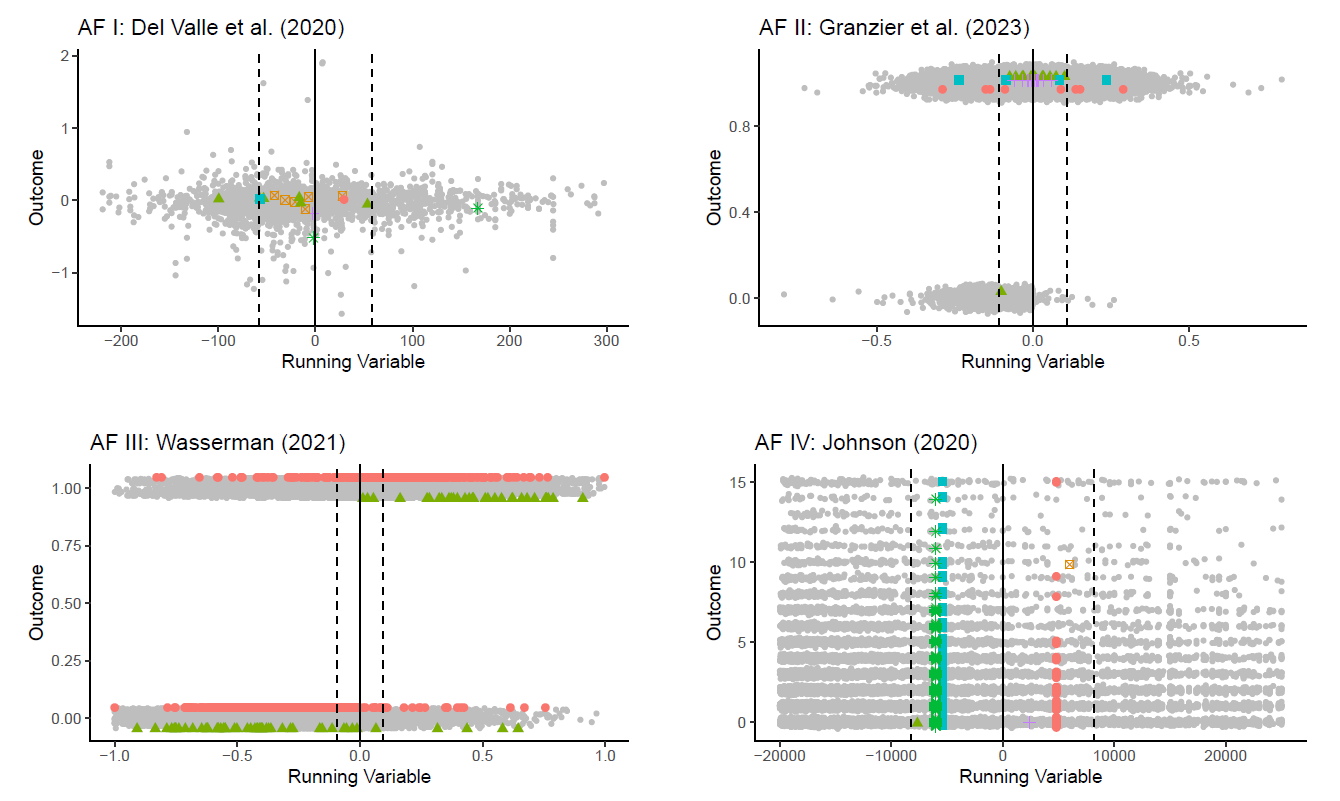}
				
	\centering		
	\caption{Visualization of cluster structures in the empirical applications.}
	\label{figure::empirical_examples_running_variable_cluster_sizes}
	\vspace{0.3cm}
	\begin{minipage}{0.95\textwidth}
		\footnotesize
		\textit{Notes:} Each point represents one observation. For discrete outcome variables, values are jittered to avoid overplotting. For selected clusters, all units are displayed in the same color and use the same marker symbol. The running variables, outcomes, and cutoffs are described in Section~\ref{subsec:empirical_applicaitons}.
	\end{minipage}
	\end{figure}


\subsubsection{Motivating Example for Asymptotic Framework~I}\label{subsubsec:valle}
\citet{delvalle2020rules} study the impact of Mexico's indexed disaster fund (Fonden) on post-disaster economic recovery.
The outcomes are constructed as changes in log night lights in the year following a disaster, measured using satellite-based night lights at the municipality level.
They leverage a fuzzy regression discontinuity design where the eligibility for disaster transfers depended on whether the realized rainfall exceeded a pre-specified cutoff.
The data contains information on municipal requests for Fonden funding in the period between 2004 and 2012.
The authors cluster the standard errors at the municipality level. The estimation window contains around 1000 municipalities with an average of 1.5 requests per municipality.

\subsubsection{Motivating Example for Asymptotic Framework~II}\label{subsubsec:granzier}
\cite{granzier2023coordination} study French two-round elections to evaluate how candidates' first-round ranking affects their second-round results. Specifically, they measure the impact of barely achieving a higher rank on remaining in the race or winning, and the running variable is the first-round vote margin between adjacent candidates. For concreteness, we focus on the effect of ranking 1st vs 2nd in the first round on the probability of running in the second round.
The data contain information on electoral races from several decades of local and parliamentary elections. The standard errors are clustered at the district level.
There are about 2300 clusters in the estimation window, with an average of 3 observations per cluster. By construction, the realizations of the running variable are symmetric around the cutoff: for every observation with running variable value $X_{gi}$, there is a corresponding observation with running variable value $-X_{gi}$. Such degenerate distributions of the running variable are allowed under our Asymptotic Framework~II.

\subsubsection{Motivating Example for Asymptotic Framework~III}\label{subsubsec:wasserman}
\citet{wasserman2021up} studies the causal effect of an electoral defeat on subsequent political participation using a close-election RD design.
The analysis focuses on first-time candidates for US state legislative offices. 
The running variable is the candidate's margin of victory, and the primary outcome of interest is whether the candidate
runs again for any state legislative office within four years of the initial run. The data covers state legislative elections over several decades across the United States, and the standard errors are clustered at the state level. 
This yields clusters with a large number of observations spread across the support of the running variable, with approximately 250 observations per state on average within a local neighborhood of the cutoff.

\subsubsection{Motivating Example for Asymptotic Framework~IV}\label{subsubsec:johnson} 
\citet{johnson2020regulation} studies the deterrence effects of a policy under which the Occupational Safety and Health Administration (OSHA) issues press releases about violations of workplace safety and health regulations that exceed a penalty threshold.
In this RD design, the running variable is the penalty amount assigned at inspection, with a discontinuity at the press-release threshold. The primary outcome is the count of violations recorded in subsequent inspections. In the dataset, the facilities are organized into ``peer groups''--groups of facilities in the same sector located within a 5 km radius of a facility where a penalty was levied---such that all facilities within a peer group share the same penalty value, while the outcome is measured at the facility level. The standard errors are clustered at the peer group level.
There are 707 peer groups of facilities, each containing roughly 16 facilities on average in a local neighborhood around the cutoff.

\begin{table}
	\centering
	\resizebox*{\textwidth}{!}{	\begin{threeparttable}
			\caption{Empirical Results}
			\label{tabelle:empirical_results}
			
			\sisetup{
				table-number-alignment = center,
				round-mode = places,
				round-precision = 3,
				round-pad = false
			}
			
			\begin{tabular}{
					l
					S[table-format=1.2]
					S[table-format=1.2] S[table-format=1.2]
					S[table-format=1.2] S[table-format=1.2] S[table-format=1.2]
					S[table-format=5]
					S[table-format=3]
					S[table-format=1.2]
					S[table-format=1.2]
				}
				\toprule
				& {Est}
				& \multicolumn{2}{c}{IID SE}
				& \multicolumn{3}{c}{Clustered SE}
				& \multicolumn{4}{c}{Cluster Sizes} \\
				\cmidrule(lr){3-4}
				\cmidrule(lr){5-7}
				\cmidrule(lr){8-11}
				& {}
				& {EHW} & {NN}
				& {Naive CNN} & {CRR} & {CNN}
				& {$n_h$}
				& {$G_h$}& {$w_{\max}$} & {$w_{\text{sum}}$ }\\
				\midrule
				AF I:  \citeauthor{delvalle2020rules} & 0.05915	& 0.02091 & 0.02319	& 0.02195 & 0.02296 & 0.02170 & 1563 & 1014	& 0.03 & 1.72 \\
				AF II:  \citeauthor{granzier2023coordination} & 0.14607 & 0.04149 & 0.03903 & 0.04030 & 0.03607 & 0.03614 & 2338 & 903 & 0.06 & 2.89 \\
				AF III: \citeauthor{wasserman2021up} & 0.50719 & 0.01528 & 0.01516 & 0.02228 & 0.02519 & 0.02415 & 12140 & 50 & 68.46 & 239.22 \\
				AF IV:  \citeauthor{johnson2020regulation} & -0.24406 & 0.11958 & 0.12292 & 0.17643 & 0.24365 & 0.24067 & 11346 & 707 & 11.04 & 55.77  \\
				\bottomrule
			\end{tabular}
			
			\begin{tablenotes}[flushleft]
				\footnotesize
				\item Notes: Column \textit{Estimate} reports the RD estimate; it can differ slightly from the one reported in the original paper since, for example, we do not include additional covariates used to improve precision in their regressions.
				\textit{EHW} and \textit{NN} standard errors do not account for clustering.
				\textit{Naive CNN} selects the nearest neighbors as in the i.i.d.\ case as implemented in the software \texttt{rdrobust}.
				\textit{CRR} is the clustered regression-residual based approach, and \textit{CNN} is our proposed standard error.
				$n_h$ and $G_h$ denote the number of observations and the number of clusters within the estimation window; $w_{\max}$ and $w_{\text{sum}}$ are the rule-of-thumb measures described in Section~\ref{subsec::rot}. The data sets are described in Section~\ref{subsec:empirical_applicaitons}.
			\end{tablenotes}
	\end{threeparttable}}
\end{table}

\subsubsection{Empirical Results}
The results for all four empirical applications are reported in Table~\ref{tabelle:empirical_results}.
Before discussing the values of the different standard errors, we first note that, according to the rule-of-thumb diagnostics reported in the last two columns of the table, the studies of \cite{delvalle2020rules} and \cite{granzier2023coordination} represent settings with relatively small and fairly balanced cluster sizes. These settings appear to fit into our Asymptotic Frameworks~I and~II well, suggesting that asymptotic normality is likely a reasonable approximation to the finite-sample distribution of the RD estimator. By contrast, in the applications studied by \cite{wasserman2021up} and \cite{johnson2020regulation}, cluster sizes are either large or markedly unbalanced. In such cases, justifying asymptotic normality requires an additional assumption on the convergence rate of the conditional variance $se^2(h)$; Assumption~\ref{ass:limit_variance} provides one illustrative sufficient condition for it to hold.

The standard errors computed under the assumption of i.i.d.\ sampling are substantially smaller than their clustered counterparts in applications with large clusters, which suggest that they fail to account for relevant within-cluster dependence. Our proposed CNN standard error is close in magnitude to the conventional CRR approach, while the Naive CNN standard error is markedly smaller in the fourth application, consistent with the discussion in Section~\ref{subsec:failure_naive_CNN}.


\section{Conclusion}
	This paper proposes a general framework for sharp RD designs with clustered data. 
	Under general high-level conditions, we establish the asymptotic normality of the local linear RD estimator, and we illustrate the high-level conditions in empirically motivated asymptotic frameworks.
	Furthermore, we develop a novel nearest-neighbors-type standard error tailored to clustered samples. Our approach is easily extendable: with minor modifications, it readily accommodates fuzzy RD and kink designs as well as settings with covariate adjustments.

	\newpage 
	\appendix
	
	\begin{center}
	\Large{Appendix}
	\end{center}

\section{Companion Clusters Selection Algorithm}\label{sec:selection_algorithms}
There are several approaches for selecting the sets $\mc R^1_g$ and $\mc R^2_g$ that satisfy our high-level conditions described in Section~\ref{subsec:SE_descitpion}. The effectiveness of a given algorithm, however, depends on the specific empirical context. To provide a concrete illustration, we propose one specific algorithm which fulfills these high-level assumptions in a broad class of empirically relevant settings.

For $\star \in \{+,-\}$, let $\mathcal X^\star_{g,h}$ denote the set of distinct realizations of the running variable in cluster $g$ within the estimation window on the respective side of the cutoff, and define $l^\star_{g,h}$ as the number of elements in $\mathcal X^\star_{g,h}$. Our proposed selection procedure is given in Algorithm~\ref{alg1}.

\begin{algorithm}[h]
	\caption{Selection of Companion Clusters $\mc R^1_g$ and $\mc R^2_g$}
	\vspace{2mm}\textbf{Inputs:} $(\mc{X}_g)_{g\in[G]}$, $R$, and $J$.\\[2mm]
	\noindent \textbf{Support Dimension Reduction:}
	Let $L \equiv \lfloor R/(4J) \rfloor$. For all $g \in [G]$ and $\star \in \{-,+\}$, define $\mc S^\star_g$ as follows:
	If $l^\star_{g,h} \leq L$ then $\mc S^\star_g = \mc X_{g,h}^\star$. Otherwise, let $\mc S^\star_g$ be the empirical quantiles of $\mc X^\star_{g,h}$ evaluated at the probabilities
	$\left\{0, \frac{1}{L-1}, \frac{2}{L-1}, \dots, 1 \right\}.$\footnotemark \\
	
	\noindent \textbf{Choice of Companion Clusters:} For $g \in [G]$:
	\begin{enumerate}[label=\textbf{Step \arabic*:}, leftmargin=*]
		\item For $\star \in \{-,+\}$ and each $x \in \mc S^\star_g$, find the $J$ closest values in $\bigcup_{\tilde g \in [G] \setminus \{g\}} \mc S^\star_{\tilde g}$ and let $r^1_g(x)$ denote the set of clusters these values belong to.
		Define:
		\[
		\mc R^1_{g} =  \bigcup_{x \in \mc S_g^- \cup \mc S_g^+ } r^1_g(x).
		\]
		
		\item For $\star \in \{-,+\}$ and each $x \in \mc S^\star_g$, find the $J$ closest values in $\bigcup_{\tilde g \in [G] \setminus \{g\} \setminus \mc R_g^1} \mc S^\star_{\tilde g}$, and let $r^2_g(x)$ denote the set of clusters these values belong to.
		Define:
		\[
		\mc R^2_{g} =  \bigcup_{x \in \mc S_g^- \cup \mc S_g^+ } r^2_g(x).
		\]
	\end{enumerate}
	\label{alg1}
\end{algorithm}
\footnotetext{\small If the running variable has mass points, we jitter the elements of $\mc S^\star_g$ by adding small independent noise.}

The above algorithm can be applied if there are at least $2JL$ clusters within the bandwidth on each side of the cutoff. This condition is in line with all the asymptotic frameworks we consider in Section~\ref{sec:AFs}, where the number of clusters in the local neighborhood of the cutoff diverges to infinity.

\begin{lemma}\label{lemma:small_lgh}
	Suppose that  Algorithm~\ref{alg1} is used. 
	Then for any $g \in [G]$, $$\#\Big\{\tilde g: g \in \mc R^1_{\tilde g} \cup \mc R_{\tilde g}^2 \Big\} \leq R.$$
\end{lemma}

\begin{proof}
	Each value in $\mc S_g^- \cup \mc S_g^+$ can be the $J$-th or closer nearest neighbor for at most $2J$ support points from other clusters. Since $|\mc S_g^- \cup \mc S_g^+| \leq 2L$, cluster $g$ can be selected as a companion cluster at most $4LJ \leq R$ times.
\end{proof}

Lemma~\ref{lemma:small_lgh} guarantees that if Algorithm~\ref{alg1} is used to define the companion clusters, then each cluster in the sample will be ``used'' at most $R$ times for the prespecified value~$R$.
This shows that the general construction described in Subsection~\ref{subsec:SE_descitpion} is feasible.

\section{Proofs of Theorems~1--3}

\subsection{Proof of Theorem~\ref{th:normality}}	
Let $W_g = \sum_{i \in I_g} w_{gi}(h)(Y_{gi} - \E[Y_{gi}|\mc X_n])$. It holds that $\E[W_g|\mc X_n] = 0$ and
\[
se^2(h) = \sum_{g \in [G]} \Var(W_g|\mc X_n).
\]
Let $\delta>2$ be as in Assumption~\ref{ass:boundedmoment}.
We will verify Lyapunov's condition by showing that
\[
A_n \equiv \frac{1}{se(h)^{\delta}} \sum_{g \in [G]} \E[ |W_g|^{\delta} |\mc X_n ] \xrightarrow{n \to \infty} 0.  
\]

First, note that
\begin{align*}
	A_n & =  \frac{ \sum_{g \in [G]} \E[ | \sum_{i \in I_g} w_{gi}(h)(Y_{gi} - \E[Y_{gi}|\mc X_n])  |^{\delta} |\mc X_n ]}{ se(h)^{\delta} } \\
	& \leq  \frac{ \sum_{g \in [G]} \E[ | \sum_{i \in I_g} |w_{gi}(h)|^{1-1/\delta} \left( |w_{gi}(h)|^{1/\delta} \left| Y_{gi} - \E[Y_{gi}|\mc X_n] \right| \right)  |^{\delta} |\mc X_n ]}{ se(h)^{\delta} }
\end{align*}	
where the first equality uses the definition of $W_g$ and the inequality follows by the triangle inequality.
Next, by H{\"o}lder inequality with exponents $\delta$ and $\delta'=\delta/(\delta-1)$, such that $1/\delta + 1/\delta'=1$, we obtain that
\begin{align*}	
	A_n	& \leq  \frac{ \sum_{g \in [G]} \E\left[ \left| \left( \sum_{i \in I_g} \left(|w_{gi}(h)|^{1-1/\delta} \right) ^{\delta'} \right)^{1/\delta'} \left( \sum_{i \in I_g} |w_{gi}(h)| \left|Y_{gi} - \E[Y_{gi}|\mc X_n]\right|^\delta \right)^{1/\delta} \right|^{\delta} |\mc X_n \right]}{ se(h)^{\delta} } \\
	& = \frac{ \sum_{g \in [G]} \E\left[ \left| \left( \sum_{i \in I_g} |w_{gi}(h)| \right)^{1/\delta'} \left( \sum_{i \in I_g} |w_{gi}(h)| \left|Y_{gi} - \E[Y_{gi}|\mc X_n]\right|^\delta \right)^{1/\delta} \right|^{\delta} |\mc X_n \right]}{ se(h)^{\delta} } \\	
	& =  \frac{ \sum_{g \in [G]}  \left( \sum_{i \in I_g} |w_{gi}(h)| \right)^{\delta/\delta'} \E\left[ \sum_{i \in I_g} |w_{gi}(h)| |Y_{gi} - \E[Y_{gi}|\mc X_n]|^\delta |\mc X_n \right]}{ se(h)^{\delta} },
\end{align*}		
where the first equality follows by the definition of $\delta'$, and the second equality uses the fact that the weights are deterministic given $\mc X_n$.

Next, using Assumption~\ref{ass:boundedmoment} that $\E[|Y_{gi} - \E[Y_{gi} | \mc X_n]|^\delta|\mc X_n]$ is uniformly bounded and the fact that $\delta/\delta' + 1=\delta$, we obtain that
\begin{align*}
	A_n  \leq  C \frac{ \sum_{g \in [G]} \left( \sum_{i \in I_g} |w_{gi}(h)| \right)^{\delta} }{ se(h)^{\delta}}.
\end{align*}	
Finally, by basic algebra,
\begin{align*}
	A_n	& \leq  C \frac{\max_{g \in [G]} \left( \sum_{i \in I_g} |w_{gi}(h)| \right)^{\delta - 2} }{ se(h)^{\delta-2}}  \frac{ \sum_{g \in [G]}  \left( \sum_{i \in I_g} |w_{gi}(h)| \right)^2}{ se(h)^{2} } \\
	& =  C \left(  \frac{\max_{g \in [G]} \left( \sum_{i \in I_g} |w_{gi}(h)| \right)^2}{se(h)^2} \right)^{\frac{\delta-2}{2}}  \frac{ \sum_{g \in [G]}  \left( \sum_{i \in I_g} |w_{gi}(h)| \right)^2}{ se(h)^{2} }.
\end{align*}
The conclusion follows by Assumption~\ref{ass:HL_cluster_size} and the fact that $\delta > 2$.
\qed

\subsection{Proof of Theorem~\ref{theorem::consistency_of_se}}
For a generic variable $D_{gi}$, let $$D^{\Delta_d}_{g i} \equiv D_{gi} - \frac{1}{| \mathcal N_{gi}^d |}\sum_{(g', i') \in  \mathcal N_{gi}^d }D_{g'i'} 
\quad \text{ for } d \in \{1,2\}.$$
Let $q_{gi}(h)= w_{gi}(h)/se(h)$
and recall that $\wh \sigma^{\sss CNN}_{g, ij} =  Y^{\Delta_1}_{g i} Y^{\Delta_2}_{gj}$. We prove Theorem~\ref{theorem::consistency_of_se} by showing that
$$ \sum_{g \in [G]} \sum_{i,j \in I_g} q_{gi}(h) q_{gj}(h)(\wh \sigma^{\sss CNN}_{g, ij} - \sigma_{g, ij}) = o_P(1).$$
Throughout the proof, we rely on the following conditions on the error term and the weights. First,
under cross-cluster independence and $\mathbb{E}[\varepsilon_{gi}\mid \mathcal X_n]=0$, it holds that
\begin{align}
	\mathbb{E}\!\left[\varepsilon_{g_1 i_1}\varepsilon_{g_2 i_2}\varepsilon_{g_3 i_3}\varepsilon_{g_4 i_4}\mid \mathcal X_n\right]=0 	\label{cond::independence}
\end{align}
unless each cluster index in $\{g_1,g_2,g_3,g_4\}$ appears at least twice.

Second,  it directly follows from Assumption~\ref{ass:HL_cluster_size} that
\begin{align}
	\sum_{g \in [G]}	\sum_{\tilde g \in [G]} \sum_{i,j \in I_g} \sum_{k,l \in I_{\tilde g}}  |q_{gi}(h) q_{gj}(h)	q_{\tilde g k}(h) q_{\tilde g l}(h)|  = O_P(1).   \label{cond::sum_everything}
\end{align} 

Third, we further note that by the triangular inequality and Assumption~\ref{ass:HL_cluster_size}, it holds that
\begin{align}
	|\sum_{g \in [G]} & \sum_{i,j,k,l \in I_g}  q_{gi}(h) q_{gj}(h)	q_{gk}(h) q_{gl}(h) | \leq \sum_{g \in [G]}	 \sum_{i,j,k,l \in I_g}   |q_{gi}(h) q_{gj}(h)	q_{gk}(h) q_{gl}(h)| \notag \\
	&  \leq \max_g  \Big\{\sum_{i,j \in I_g} |q_{gi}(h) q_{gj}(h)| \Big\} \sum_{g \in [G]}	 \sum_{k,l \in I_g}  	|q_{gk}(h) q_{gl}(h)|  = o_P(1).  \label{cond::weightsqaure}
\end{align} 

In the following derivations, we write $C$ for a generic positive constant whose value might differ
between equations. 
For $d \in \{1,2\}$, we let 
$$Y^{\Delta_d}_{g i} =    \left( \mu(X_{gi}) - \frac{1}{|\mathcal N_{gi}^d|} \sum_{(g_d, i_d) \in \mathcal N_{gi}^d} \mu(X_{g_d i_d})  \right)+ \left(\varepsilon_{g i} - \frac{1}{|\mathcal N_{gi}^d|} \sum_{(g_d, i_d) \in \mathcal N_{gi}^d}\varepsilon_{g_d i_d}\right) \equiv  \mu^{\Delta_d}_{g i} + \varepsilon^{\Delta_d}_{g i}.$$

First, we show that the standard error is asymptotically unbiased. It holds that 
\begin{align*}
	\mathbb E \bigg[	\frac{\wh{se}^2_{\sss CNN}(h)}{se^2(h)} -1| \mathcal X_n\bigg] 
	& =  \mathbb E \bigg[\sum_{g \in [G]} \sum_{i, j \in I_g } q_{gi}(h) q_{gj}(h)  \left( Y^{\Delta_1}_{g i} Y^{\Delta_2}_{g j} - \sigma_{g, ij} \right) |\mathcal X_n \bigg]\\
	& =	\sum_{g \in [G]} \sum_{i, j \in I_g}  q_{gi}(h) q_{gj}(h)  \mu^{\Delta_1}_{g i}\mu^{\Delta_2}_{g j} \\
	&  \quad +	\sum_{g \in [G]} \sum_{i, j \in I_g } q_{gi}(h) q_{gj}(h)  \mathbb E [  \varepsilon^{\Delta_1}_{g i}   |\mathcal X_n]  \mu^{\Delta_2}_{g j}   \\
	& \quad +	\sum_{g \in [G]} \sum_{i, j \in I_g } q_{gi}(h) q_{gj}(h)  \mathbb E [  \varepsilon^{\Delta_2}_{g j}   |\mathcal X_n]   \mu^{\Delta_1}_{g i}  \\
	& \quad +	\sum_{g \in [G]} \sum_{i, j \in I_g } q_{gi}(h) q_{gj}(h)  \mathbb E [ \varepsilon^{\Delta_1}_{g i}  \varepsilon^{\Delta_2}_{gj}    - \sigma_{g,ij} |\mathcal X_n].
\end{align*} 	
Since $\varepsilon_{gi}$  is conditional mean zero for all $g \in [G]$ and $i\in I_g$ and by independence between clusters and construction of the sets $\mathcal N_{gi}^d$, it holds that for all pairs $i,j \in I_g$, 
\begin{align*}
	\mathbb E [  \varepsilon^{\Delta_1}_{g i}  \varepsilon^{\Delta_2}_{gj}|\mathcal X_n] & = 	\mathbb E [ \left(\varepsilon_{g i} - \frac{1}{|\mathcal N_{gi}^1|} \sum_{(g_1, i_1) \in \mathcal N_{gi}^1}\varepsilon_{g_1 i_1}\right)  \left(\varepsilon_{g j} - \frac{1}{|\mathcal N_{gj}^2|} \sum_{(g_2,j_2) \in \mathcal N_{gj}^2}\varepsilon_{g_2 j_2}\right)  |\mathcal X_n]\\
	&   = 	\mathbb E [ \varepsilon_{g i} \varepsilon_{g j} |\mathcal X_n]  = \sigma_{g,ij}.
\end{align*}

We further note that
\begin{align*}
	\left\lvert \mathbb E \bigg[	\frac{\wh se^2_{\sss CNN}(h)}{se^2(h)} -1| \mathcal X_n\bigg] \right\rvert & 	= \left\lvert \sum_{g \in [G]}\sum_{i, j \in I_g } 	q_{gi}(h) q_{gj}(h) \mu^{\Delta_1}_{g i} \mu^{\Delta_2}_{g j} \right\rvert \\
	& \leq C
	\max_{g \in [G]}  \max_{ \substack{i \in I_g \\ w_{gi}(h) \neq 0}}   \max_{(g', i') \in  \mathcal N_{gi}^1 \cup \mathcal N_{gi}^2} |X_{gi} - X_{g' i'}|^2 \sum_{g \in [G]} \sum_{i, j \in I_g } |q_{gi}(h) q_{gj}(h)| \\
	& = o_P(1).
\end{align*}

The first inequality follows from the L-Lipschitz continuity of $\mu$ and the triangular inequality. The last line follows as $\sum_{g \in [G]} \sum_{i, j \in I_g } |q_{gi}(h) q_{gj}(h)| = O_P(1)$  and by Assumption~\ref{ass:HL_closest_best_neighbor}. We have shown that the standard error is asymptotically unbiased.

Second, we study the conditional variance of $\wh {se}_{\sss CNN}^2(h)$.  We consider the following decomposition
\begin{align*}
	& \frac{\wh {se}_{\sss CNN}^2(h)-  \mathbb E [\wh {se}_{\sss CNN}^2(h) |\mathcal X_n]}{se^2(h)} 
	= \sum_{g \in [G]} \sum_{i,j \in I_g} q_{gi}(h) q_{gj}(h) \left( (\varepsilon^{\Delta_1}_{g i}  + \mu^{\Delta_1}_{g i} ) (\varepsilon^{\Delta_2}_{g j} + \mu^{\Delta_2}_{g j} )  -  \mathbb E [\wh \sigma^{\sss CNN}_{g, ij} |\mathcal X_n] \right) \\
	&  = \sum_{g \in [G]} \sum_{i,j \in I_g} q_{gi}(h) q_{gj}(h) \times \Big( (\varepsilon_{gi} \varepsilon_{gj}-  \mathbb E [\varepsilon_{gi}  \varepsilon_{gj}  |\mathcal X_n] )  -	\left(\frac{1}{|\mathcal N_{gi}^1|} \sum_{(g_1, i_1) \in \mathcal N_{gi}^1}\varepsilon_{g_1 i_1}\right) \varepsilon_{gj}  \\
	& \qquad \qquad - \left(\frac{1}{|\mathcal N_{gj}^2|} \sum_{(g_2,j_2) \in \mathcal N_{gj}^2}\varepsilon_{g_2 j_2}\right)   \varepsilon_{gi} + \left(\frac{1}{|\mathcal N_{gj}^2|} \sum_{(g_2,j_2) \in \mathcal N_{gj}^2}\varepsilon_{g_2 j_2}\right)    	\left(\frac{1}{|\mathcal N_{gi}^1|} \sum_{(g_1, i_1) \in \mathcal N_{gi}^1}\varepsilon_{g_1 i_1}\right)  \\
	& \qquad \qquad +  \varepsilon_{gi}\mu^{\Delta_2}_{g j} - 
	\left(\frac{1}{|\mathcal N_{gi}^1|} \sum_{(g_1, i_1) \in \mathcal N_{gi}^1}\varepsilon_{g_1 i_1}\right) \mu^{\Delta_2}_{g j}
	+ \varepsilon_{gj} \mu^{\Delta_1}_{g i} 
	- \left(\frac{1}{|\mathcal N_{gj}^2|} \sum_{(g_2,j_2) \in \mathcal N_{gj}^2}\varepsilon_{g_2 j_2}\right)  \mu^{\Delta_1}_{g i}   \Big)  \\
	& \equiv A_1  - A_2 - A_3 + A_4 + A_5 - A_6 + A_7 - A_8.
\end{align*}
It is easy to see that these eight terms  are all mean zero conditional on $\mathcal X_n$. It thus suffices
to show that their second moments converge to zero.
We will consider each of the terms $A_1, \dots, A_8$ separately. 

We start with $A_1$. It holds that
\begin{align*}
	&\Var(A_1| \mathcal X_n)\\
	& =  \sum_{g \in [G]}  \sum_{\tilde g \in [G]}  \sum_{i,j \in I_g} \sum_{l,k  \in I_{\tilde g}} q_{gi}(h) q_{gj}(h) q_{\tilde g  k}(h) q_{\tilde g l}(h)  \mathbb E [ (\varepsilon_{gi} \varepsilon_{gj} -  \mathbb E [\varepsilon_{gi}  \varepsilon_{gj}  |\mathcal X_n])  (\varepsilon_{\tilde gl} \varepsilon_{\tilde gk} -  \mathbb E [\varepsilon_{\tilde gl}  \varepsilon_{\tilde g k}  |\mathcal X_n])   |\mathcal X_n]\\
	& =  \sum_{g \in [G]}   \sum_{i,j, l,k \in I_g}  q_{gi}(h) q_{gj}(h) q_{g k}(h) q_{gl}(h)  \mathbb E [ (\varepsilon_{gi} \varepsilon_{gj} -  \mathbb E [\varepsilon_{gi}  \varepsilon_{gj}  |\mathcal X_n])  (\varepsilon_{g l} \varepsilon_{ g k} -  \mathbb E [\varepsilon_{ g l}  \varepsilon_{g k}  |\mathcal X_n])   |\mathcal X_n]\\
	& \leq C   \sum_{g \in [G]}   \sum_{i,j, l,k \in I_g}  |q_{gi}(h) q_{gj}(h) q_{g k}(h) q_{g l}(h)| = o_P(1)
\end{align*}
The second equality follows as the units are independent across clusters and by Condition~\ref{cond::independence}. The first inequality follows by boundedness of fourth moments of the error term and the last equality follows from Condition~\ref{cond::weightsqaure}.

We now consider $A_2$
\begin{align*}
	&	\Var(A_2| \mathcal X_n) 
	\\
	& =  \sum_{g, \tilde g \in [G]}  \sum_{i,j \in I_g} \frac{1}{|\mathcal N_{gi}^1|} \sum_{(g_1, i_1) \in \mathcal N_{gi}^1} \sum_{l,k  \in I_{\tilde g}} \frac{1}{|\mathcal N_{\tilde gk}^1|} \sum_{(\tilde g_{1}, k_1) \in \mathcal N_{\tilde g k }^1} q_{gi}(h) q_{gj}(h) q_{\tilde gk}(h) q_{\tilde g l}(h)  \mathbb E [ \varepsilon_{g_{1} i_1} \varepsilon_{gj}\varepsilon_{\tilde g_{1} k_1} \varepsilon_{\tilde g l} |\mathcal X_n]
\end{align*}
Based on the logic of Condition~\ref{cond::independence}, each of those terms of the sums  are nonzero if either $(g_1=\tilde g_1 \,\&\, g = \tilde g)$ or $(g_1=\tilde g \,\&\, g = \tilde g_1)$. By  the boundedness of the conditional expectations of the first four moments of the error term, it then follows that
\begin{align*}
	&	\Var(A_2| \mathcal X_n) 
	\leq  C \sum_{g \in [G]}   \sum_{i,j, l,k \in I_g} 	 \sum_{(g_1, i_1)  \in \mathcal N_{gi}^1}
	\sum_{(\tilde g_{1}, k_1) \in \mathcal N_{\tilde g k}^1} \frac{1}{|\mathcal N_{gi}^1|}	\frac{1}{|\mathcal N_{\tilde gk}^1|}  \1{g_1 = \tilde g_1}  |q_{gi}(h) q_{gj}(h) q_{\tilde g k}(h) q_{\tilde g l}(h) | \\
	& \; +C  \sum_{g, \tilde g \in [G]}   \sum_{i,j \in I_g} \sum_{l,k  \in I_{\tilde g}} 	\frac{1}{|\mathcal N_{gi}^1|} \sum_{(g_1, i_1)  \in  \mathcal N_{gi}^1} 
	\frac{1}{|\mathcal N_{\tilde gk}^1|} \sum_{(\tilde g_{1}, k_1) \in\mathcal N_{\tilde g k}^1}
	\1{g_1 = \tilde g \,\&\,  \tilde g_1 = g } |q_{gi}(h) q_{gj}(h) q_{\tilde g k}(h) q_{\tilde g l}(h)|\\
	& \equiv H_1 +H_2
\end{align*}
Furthermore, by Condition~\ref{cond::weightsqaure}, 
\begin{align*}
	H_1	& \leq 	C \sum_{g \in [G]}   \sum_{i,j,k,l\in I_g}  |q_{g i}(h) q_{g j}(h) |	\frac{1}{|\mathcal N_{gi}^1|} \sum_{(g_1, i_1)  \in \mathcal N_{gi}^1}
	\frac{1}{|\mathcal N_{gk}^1|} \sum_{(\tilde g_{1}, k_1) \in \mathcal N_{g k}^1} | q_{g k}(h) q_{g l}(h) | \\
	&  \leq C	\sum_{g \in [G]}   \sum_{i,j, k,l \in I_g}  |q_{g i}(h) q_{g j}(h) q_{g k}(h) q_{g l}(h) | = o_P(1).
\end{align*}
We further note that as by construction, Equation~\ref{equ:max_R}, each cluster is contributes neighbors only to units of R other clusters
\begin{align*}
	H_2 
	&\le C \sum_{g \in [G]} \sum_{i,j \in I_g}
	\frac{|q_{gi}(h) q_{gj}(h)|}{|\mathcal N_{gi}^1|}
	\sum_{(g_1, i_1) \in \mathcal N_{gi}^1} 
	\sum_{\substack{
			\tilde g \in [G] \\
			g \in \mathcal R_{\tilde g}^1 \\
			\tilde g \in \mathcal R_g^1
	}}
	\Bigg(
	\sum_{l,k \in I_{\tilde g}}
	\frac{|q_{\tilde g k}(h) q_{\tilde g l}(h)|}{|\mathcal N_{\tilde g k}^1|}
	\sum_{(\tilde g_1, k_1) \in \mathcal N_{\tilde g k}^1}
	\mathbf{1}\{ g_1 = \tilde g,\; \tilde g_1 = g \}
	\Bigg) \\
	&\le C 
	\left(
	\max_{\tilde g \in [G]}
	\sum_{l,k \in I_{\tilde g}} 
	|q_{\tilde g k}(h) q_{\tilde g l}(h)|
	\right)
	\sum_{g \in [G]} \sum_{i,j \in I_g} 
	|q_{gi}(h) q_{gj}(h)| \\
	&= o_P(1).
\end{align*}
where the last inequality follows by the restriction of Equation~\eqref{equ:max_R} as each cluster is used as a companion cluster at most $R$ times. The last line follow from condition~\eqref{cond::weightsqaure}.
It follows that $\Var(A_2| \mathcal X_n) = o_P(1).$ Using the same arguments, it also holds that $\Var(A_3|\mathcal X_n) = o_P(1)$.

We now consider $A_4$.
\begin{align*}
	& \Var(A_4| \mathcal X_n)   =  \sum_{g \in [G]}  \sum_{\tilde g \in [G]}  \sum_{i,j \in I_g} \sum_{l,k  \in I_{\tilde g}}
	\frac{1}{|\mathcal N_{gj}^2|} \sum_{(g_2,j_2) \in \mathcal N_{gj}^2}
	\frac{1}{|\mathcal N_{gi}^1|} \sum_{(g_1, i_1) \in \mathcal N_{gi}^1}
	\frac{1}{|\mathcal N_{\tilde gl}^2|} \sum_{(\tilde g_2, l_2) \in \mathcal N_{\tilde g l}^2}
	\frac{1}{|\mathcal N_{\tilde gk}^1|}  \sum_{(\tilde g_1, k_1) \in \mathcal N_{\tilde gk}^1}
	\\
	&\qquad \qquad  \cdot  q_{gi}(h) q_{gj}(h) q_{\tilde g k}(h) q_{\tilde g l}(h)  \mathbb E [ \varepsilon_{g_2 j_2}\varepsilon_{g_1 i_1}
	\varepsilon_{\tilde g_2 l_2}    	\varepsilon_{\tilde g_1  k_1}|\mathcal X_n]\\
\end{align*}
Based on the logic of Condition~\ref{cond::independence}, each of those terms of the sums  are nonzero if either $(g_1=\tilde g_1 \,\&\, g_2 = \tilde g_2)$ or $(g_1=\tilde g_2 \,\&\, g_2 = \tilde g_1)$. By  the boundedness of the conditional expectations of the first four moments of the error term, it then follows that
\begin{align*}
	& \Var(A_4| \mathcal X_n)   \leq  \sum_{g \in [G]}  \sum_{\tilde g \in [G]}  \sum_{i,j \in I_g} \sum_{l,k  \in I_{\tilde g}}
	\frac{1}{|\mathcal N_{gj}^2|} \sum_{(g_2,j_2) \in \mathcal N_{gj}^2}
	\frac{1}{|\mathcal N_{gi}^1|} \sum_{(g_1, i_1) \in \mathcal N_{gi}^1}
	\frac{1}{|\mathcal N_{\tilde gl}^2|} \sum_{(\tilde g_2, l_2) \in \mathcal N_{\tilde g l}^2}
	\frac{1}{|\mathcal N_{\tilde gk}^1|}  \sum_{(\tilde g_1, k_1) \in \mathcal N_{\tilde gk}^1}
	\\
	&\qquad   \cdot   \1{g_1=\tilde g_1 \,\&\, g_2 = \tilde g_2} |q_{gi}(h) q_{gj}(h) q_{\tilde g k}(h) q_{\tilde g l}(h)|  \\ 
	& \quad + \sum_{g \in [G]}  \sum_{\tilde g \in [G]}  \sum_{i,j \in I_g} \sum_{l,k  \in I_{\tilde g}}
	\frac{1}{|\mathcal N_{gj}^2|} \sum_{(g_2,j_2) \in \mathcal N_{gj}^2}
	\frac{1}{|\mathcal N_{gi}^1|} \sum_{(g_1, i_1) \in \mathcal N_{gi}^1}
	\frac{1}{|\mathcal N_{\tilde gl}^2|} \sum_{(\tilde g_2, l_2) \in \mathcal N_{\tilde g l}^2}
	\frac{1}{|\mathcal N_{\tilde gk}^1|}  \sum_{(\tilde g_1, k_1) \in \mathcal N_{\tilde gk}^1}
	\\
	&\qquad   \cdot \1{g_1=\tilde g_2 \,\&\, g_2 = \tilde g_1}
	|q_{gi}(h) q_{gj}(h) q_{\tilde g k}(h) q_{\tilde g l}(h)|  \\ 
	& \equiv H_3+H_4
\end{align*}

The first inequality follows from Condition~\ref{cond::independence} and the boundedness of the conditional expectations of the first four moments of the error term. 
We further note that
\begin{align*}
	&H_3  =	 \sum_{g \in [G]}  \sum_{i,j \in I_g}
	\frac{1}{|\mathcal N_{gj}^2|} \sum_{(g_2,j_2) \in \mathcal N_{gj}^2}
	\frac{1}{|\mathcal N_{gi}^1|} \sum_{(g_1, i_1) \in \mathcal N_{gi}^1}
	|q_{gi}(h) q_{gj}(h)| \\
	& \cdot  \left(
	\sum_{ \substack{\tilde g \in [G]:\\ \mathcal R_g^1 \cap \mathcal R_{\tilde g}^1 \neq \emptyset \\ \mathcal R_g^2 \cap \mathcal R_{\tilde g}^2 \neq \emptyset}} 
	\sum_{l,k  \in I_{\tilde g}}
	\frac{1}{|\mathcal N_{\tilde gl}^2|} \sum_{(\tilde g_2, l_2) \in \mathcal N_{\tilde g l}^2}
	\frac{1}{|\mathcal N_{\tilde gk}^1|}  \sum_{(\tilde g_1, k_1) \in \mathcal N_{\tilde gk}^1}   | q_{\tilde g k}(h) q_{\tilde g l}(h) |  \1{g_1 = \tilde g_1 \,\&\,  g_2 = \tilde g_2} \right) \\
	& \leq C \bigg( \max_{\tilde g \in [G]}   \sum_{l,k  \in I_{\tilde g}}
	\frac{1}{|\mathcal N_{\tilde gl}^2|} \sum_{(\tilde g_2, l_2) \in \mathcal N_{\tilde g l}^2}
	\frac{1}{|\mathcal N_{\tilde gk}^1|}  \sum_{(\tilde g_1, k_1) \in \mathcal N_{\tilde gk}^1}   | q_{\tilde g k}(h) q_{\tilde g l}(h) |  \1{g_1 = \tilde g_1 \,\&\,  g_2 = \tilde g_2} \bigg) \\
	&\quad \cdot	\sum_{g \in [G]}  \sum_{i,j \in I_g}
	\frac{1}{|\mathcal N_{gj}^2|} \sum_{(g_2,j_2) \in \mathcal N_{gj}^2}
	\frac{1}{|\mathcal N_{gi}^1|} \sum_{(g_1, i_1) \in \mathcal N_{gi}^1}
	|q_{gi}(h) q_{gj}(h)| \\
	&  = o_P(1)
\end{align*}
where the inequality follows from the condition~\eqref{equ:max_R}. If each cluster is used as a companion cluster at most $R$ times, cluster can share also only a bounded number of common companion clusters.  
The last equality follows from Condition~\ref{cond::weightsqaure}. By the same arguments, it holds that
$ H_4  = o_P(1).$
It follows that $\Var(A_4| \mathcal X_n) = o_P(1)$.

We now consider $A_5$. It holds that
\begin{align*}
	\Var(A_5| \mathcal X_n) & = \sum_{g \in [G]}  \sum_{\tilde g \in [G]}  \sum_{i,j \in I_g} \sum_{l,k  \in I_{\tilde g}} q_{gi}(h) q_{gj}(h) q_{\tilde g k}(h) q_{\tilde g l}(h) 
	\mathbb E [ \varepsilon_{gi} \varepsilon_{\tilde g k} | \mathcal X_n]  \mu^{\Delta_2}_{g j} \mu^{\Delta_2}_{\tilde gl}\\
	& = \sum_{g \in [G]}   \sum_{i,j,l,k \in I_g}  q_{gi}(h) q_{gj}(h) q_{gk}(h) q_{gl}(h) 
	\mathbb E [ \varepsilon_{gi} \varepsilon_{gk} | \mathcal X_n]  \mu^{\Delta_2}_{gj} \mu^{\Delta_2}_{ gl}\\
	& \leq C \max_{g \in [G]} \max_{j \in I_g} (\mu^{\Delta_2}_{gj})^2  \sum_{g \in [G]}   \sum_{i,j,l,k \in I_g}  |q_{gi}(h) q_{gj}(h) q_{gk}(h) q_{gl}(h)| \\
	& =o_P(1).
\end{align*}

The second equality follows from Condition~\ref{cond::independence}. The inequality follows as the first four moments of the error terms are bounded. The last step follows from Condition~\ref{cond::weightsqaure} and   as $\max_{g \in [G]} \max_{j \in I_g} (\mu^{\Delta_2}_{gj})^2= o_P(1)$ by Assumption~\ref{ass:HL_closest_best_neighbor} and Lipschitz continuity of $\mu$. 

We now consider $A_6$. It holds that
\begin{align*}
	\Var(A_6| \mathcal X_n) & = \sum_{g \in [G]}  \sum_{\tilde g \in [G]}  \sum_{i,j \in I_g} \sum_{l,k  \in I_{\tilde g}} \frac{1}{|\mathcal N_{gi}^1|} \sum_{(g_1, i_1) \in \mathcal N_{gi}^1} \frac{1}{|\mathcal N_{\tilde g k}^1|} \sum_{(\tilde g_1, k_1) \in \mathcal N_{\tilde g k }^1} \\
	& \quad \cdot \left( q_{gi}(h) q_{gj}(h) q_{\tilde gk}(h) q_{\tilde g l}(h)  \right)
	\mathbb E [ \varepsilon_{g_1i_1}  \varepsilon_{\tilde g_1k_1}  | \mathcal X_n]  \mu^{\Delta_2}_{gj} \mu^{\Delta_2}_{\tilde gl}\\
	& \leq C \max_{g \in [G]} \max_{j \in I_g} (\mu^{\Delta_2}_{gj})^2 \sum_{g \in [G]}   \sum_{i,j \in I_g} |q_{gi}(h) q_{gj}(h) \frac{1}{|\mathcal N_{gi}^1|} \sum_{(g_1, i_1) \in \mathcal N_{gi}^1} \sum_{\substack{\tilde g \in [G] \\ \mathcal R_g^1 \cap  \mathcal R_{\tilde g}^1 \neq \emptyset}}   \sum_{l,k  \in I_{\tilde g}}  \\
	& \qquad \cdot \left(\frac{1}{|\mathcal N_{\tilde g k}^1|} \sum_{(\tilde g_1, k_1) \in \mathcal N_{\tilde g k }^1}  \1{g_1 = \tilde g_1 }   |q_{\tilde g k}(h) q_{\tilde g l}(h)| \right)\\
	& = o_P(1).
\end{align*}
The first inequality follows as the first four moments of the error terms are bounded and by Condition~\ref{cond::independence}.   The last equality follows as $ \max_{g \in [G]} \max_{j \in I_g} (\mu^{\Delta_2}_{gj})^2= o_P(1)$ by Assumption~\ref{ass:HL_closest_best_neighbor} and from Condition~\ref{cond::sum_everything}.

One can show that $\Var(A_7| \mathcal X_n)=o_P(1)$ using the same arguments as in the discussion of $A_5$. Similarly, one can show that $\Var(A_8| \mathcal X_n)=o_P(1)$ using the same arguments as in the discussion of $A_6$.  This reasoning concludes the proof.  \hfill$\square$

\subsection{Proof of Theorem~\ref{theorem:RR_SE}}
Define the infeasible version of $\wh{se}_{\sss CRR}^2(h)$ using the true residuals $\varepsilon_{gi} = Y_{gi} - \mu(X_{gi})$,
\[
\wt{se}_{\sss CRR}^2(h) = \sum_{g \in [G]} \sum_{i,j \in I_g} w_{gi}(h)w_{gj}(h) \varepsilon_{gi} \varepsilon_{gj}.
\]	
First, we show that $\wh{se}_{\sss CRR}^2(h)$ and $\wt{se}_{\sss CRR}^2(h)$ are first-order equivalent,
\begin{equation}\label{eq:EHW_step1}
	\frac{\wh{se}_{\sss CRR}^2(h) - \wt{se}_{\sss CRR}^2(h)}{se^2(h)} = o_P(1).
\end{equation}
To prove this claim, observe that
\begin{align*}
	\hat\varepsilon_{gi} \hat\varepsilon_{gj} - \varepsilon_{gi} \varepsilon_{gj}
	& = \left( Y_{gi}Y_{gj} - \wh\mu(X_{gi})Y_{gj} - Y_{gi}\wh\mu(X_{gj}) + \wh\mu(X_{gi})\wh\mu(X_{gj})   \right) \\
	& \quad -  \left( Y_{gi}Y_{gj} - \mu(X_{gi})Y_{gj} - Y_{gi}\mu(X_{gj}) + \mu(X_{gi})\mu(X_{gj})   \right) \\
	%
	%
	%
	& = (\mu(X_{gi}) - \wh\mu(X_{gi}))(Y_{gj}- \mu(X_{gj})) + (\mu(X_{gj}) - \wh\mu(X_{gj})) (Y_{gi} - \mu(X_{gi})) \\
	& \quad  + \left(\wh\mu(X_{gi}) - \mu(X_{gi}) \right) \left(\wh\mu(X_{gj}) -  \mu(X_{gj}) \right).
\end{align*}
It follows that
\begin{align*}
	\wh{se}_{\sss CRR}^2(h) - \wt{se}_{\sss CRR}^2(h)
	& = \sum_{g \in [G]} \sum_{i,j \in I_g} w_{gi}(h)w_{gj}(h) \left(\hat\varepsilon_{gi} \hat\varepsilon_{gj} - \varepsilon_{gi} \varepsilon_{gj} \right) \\
	& = 2 \sum_{g \in [G]} \sum_{i,j \in I_g} w_{gi}(h)w_{gj}(h)  (\mu(X_{gi}) - \wh\mu(X_{gi}))(Y_{gj}- \mu(X_{gj}))  \\
	& \quad  + \sum_{g \in [G]} \sum_{i,j \in I_g} w_{gi}(h)w_{gj}(h) \left(\wh\mu(X_{gi}) - \mu(X_{gi}) \right) \left(\wh\mu(X_{gj}) -  \mu(X_{gj}) \right) . 
\end{align*}
We can decompose $\wh\mu(X_{gi}) - \mu(X_{gi})$ as follows:
\begin{align*}
	\wh\mu(X_{gi}) - \mu(X_{gi}) & = \left( (\hat b_0^+ - b_0^+) + (\hat b_1^+ - b_1^+)X_{gi} \right)\1{0 \leq X_{gi}} \\
	& \quad + \left( (\hat b_0^- - b_0^-) + (\hat b_1^- - b_1^-)X_{gi}  \right) \1{ X_{gi} < 0 } + B_{gi},
\end{align*}
where $B_{gi}$ is the remainder from a Taylor expansion of $\mu$ on the respective side of the cutoff that depends only on $X_{gi}$ and $\max_{g \in [G]}\max_{i \in I_g: |X_{gi}| \leq h}|B_{gi}| = O(h^2).$
It follows that
\begin{align*}
	\wh{se}_{\sss CRR}^2 - \wt{se}_{\sss CRR}^2
	& =  2 \sum_{\star \in \{+,-\}}   (\hat b_0^\star - b_0^\star) \sum_{g \in [G]} \sum_{i,j \in I_g} (-1)^{\1{\star=-}} w^\star_{gi}(h)w_{gj}(h)   (Y_{gj}- \mu(X_{gj})) \\
	& \quad + 2 \sum_{\star \in \{+,-\}}  h (\hat b_1^\star - b_1^\star) \sum_{g \in [G]} \sum_{i,j \in I_g} (-1)^{\1{\star=-}}w^\star_{gi}(h)w_{gj}(h) (X_{gi}/h)  (Y_{gj}- \mu(X_{gj})) \\
	&  \quad + 2 \sum_{g \in [G]} \sum_{i,j \in I_g} w_{gi}(h)w_{gj}(h) B_{gi}  (Y_{gj}- \mu(X_{gj})) \\
	& \quad + \sum_{g \in [G]} \sum_{i,j \in I_g} w_{gi}(h)w_{gj}(h) \left(\wh\mu(X_{gi}) - \mu(X_{gi}) \right) \left(\wh\mu(X_{gj}) -  \mu(X_{gj}) \right)  \\
	& \equiv 2A_1 + 2A_2 + 2A_3 + A_4.
\end{align*}	
We begin by studying the first two terms. Let
\begin{align*}
	T^{\star}_l = \sum_{g \in [G]} \sum_{i,j \in I_g} w^\star_{gi}(h)w_{gj}(h) (X_{gi}/h)^l (Y_{gj}- \mu(X_{gj})).
\end{align*}
Note that, $\E[T_l^{\star} | \mc X_n] = 0$ and, using the fact that $\Var(\varepsilon_{gi}|\mc X_g)$ is uniformly bounded,
\begin{align*}
	\Var(T_l^{\star} | \mc X_n) 
	& = \sum_{g \in [G]} \E\left[ \left(\sum_{i,j \in I_g} w^\star_{gi}(h) w_{gj}(h)   (X_{gi}/h)^{l} (Y_{gj}- \mu(X_{gj})) \right)^2 \Big| \mc X_g \right]  \\
	& \leq C \sum_{g \in [G]} \left(\sum_{i,j \in I_g} |w^\star_{gi}(h) w_{gj}(h)| \right)^2 \\
	& \leq C \max_{g \in [G]} \sum_{i,j \in I_g}|w^\star_{gi}(h)w_{gj}(h)| \sum_{g \in [G]} \sum_{i,j \in I_g} |w^\star_{gi}(h)w_{gj}(h)|.
\end{align*}
By Assumption~\ref{ass:HL_cluster_size}, it then follows that
\[
T_l^{\star}  / se^2(h) = o_P(1),
\]
and, in consequence, $A_1/se^2(h) = o_P(1)$ and  $A_2/se^2(h) = o_P(1)$, using $\hat b_0^\star - b_0^\star = o_P(1)$ and $h(\hat b_1^\star - b_1^\star) = o_P(1)$. By the same reasoning, $A_3/se^2(h) = O_P(h^2) = o_P(1)$. The last term is bounded as follows:
\[
|A_4| \leq\max_{g \in [G]}\max_{i \in I_g: |X_{gi}| \leq h} \left(\wh\mu(X_{gi}) - \mu(X_{gi})\right)^2 \sum_{g \in [G]} \sum_{i,j \in I_g} |w_{gi}(h)w_{gj}(h)|,
\]
such that $A_4/se^2(h) = o_P(1)$.

Second, we show that
\begin{equation}\label{eq:EHW_step2}
	\frac{\wt{se}_{\sss CRR}^2(h) - se^2(h)}{se^2(h)} = o_P(1).
\end{equation}
To see that, note that
\[
\E[\wt{se}_{\sss CRR}^2(h)|\mc X_n] - se^2(h)  =  \sum_{g \in [G]} \sum_{i,j \in I_g} w_{gi}(h)w_{gj}(h) \E\left[ \varepsilon_{gi} \varepsilon_{gj} - \sigma_{g,ij} | \mc X_g \right] =0,
\]
and
\begin{align*}
	\Var\big(\wt{se}_{\sss CRR}^2(h) |\mc X_n\big) 
	& \leq \sum_{g \in [G]} \left( \sum_{i,j \in I_g}| w_{gi}(h)w_{gj}(h) |\right)^2 \max_{i,j \in I_g}\Var(\varepsilon_{gi} \varepsilon_{gj} | \mc X_g )\\
	& \leq C  \max_{g \in [G]} \sum_{i,j \in I_g}|w_{gi}(h)w_{gj}(h)| \sum_{g \in [G]} \sum_{i,j \in I_g} |w_{gi}(h)w_{gj}(h)| \\
	& = o_P(1).
\end{align*}
The proof is concluded by combining \eqref{eq:EHW_step1} and~\eqref{eq:EHW_step2}.\qed

\section{Proofs for Section~4}

\subsection{Additional Notation and Lemmas}\label{A:sec:weights}
Let  $k^+(v)=k(v)\1{0 \leq v}$,  $k^-(v)=k(v)\1{v<0}$, and $k_h^\star(v)=k^\star(v/h)/h$ for $\star \in \{+,-\}$.
The local linear RD estimator is defined as
\begin{align*}
	\wh\tau(h) &= \sum_{g \in [G]} \sum_{i \in I_g} w_{gi}(h) Y_i,\quad 
	w_{gi}(h) = w_{gi}^+(h) - w_{gi}^-(h), \\
	w_{gi}^\star(h) &= e_{1}^\top \left(\sum_{g \in [G]} \sum_{i \in I_g} k^\star_h(X_i) \widetilde{X}_{gi} \widetilde{X}_{gi}^\top\right)^{-1} k^\star_h(X_{gi}) \widetilde{X}_{gi} \quad\text{ for } \star \in \{+,-\},
\end{align*}
where $\widetilde{X}_{gi}=(1,X_{gi})^\top$.
The local linear weights can be further expressed as
\[
w^\star_{gi}(h) = \frac{ \frac{1}{n} k^\star_h(X_{gi}) \left(S^\star_{n,2}   - S^\star_{n,1}  (X_{gi}/h) \right)  }{ S^\star_{n,2} S^\star_{n,0}- (S^\star_{n,1})^2 }, \quad S^\star_{n,l} = \frac{1}{n} \sum_{g\in[G]}\sum_{i \in I_g} k^\star_h(X_{gi}) (X_{gi}/h)^l.
\]

To prove the results in Section~\ref{sec:AFs}, we first establish two technical lemmas. The first lemma provides basic convergence results for $S_n$ and other kernel-weighted sums that are used in
all four asymptotic frameworks.
For $l,m \in \mathbb{N}_0$ and $\star, \diamond \in \{+,-\} $,
\begin{align*}
	& T^\star_{n,l} = \frac{1}{n^2} \sum_{g \in [G]} \sum_{i \in I_g} (k^\star_h(X_{gi}))^2 (X_{gi}/h)^l, \\
	& U^{\star\diamond}_{n,l,m} = \frac{1}{n^2} \sum_{g \in [G]} \sum_{\substack{ i \neq j  \\ i,j \in I_g} } k^{\star}_h(X_{gi})k^{\diamond}_h(X_{gj})(X_{gi}/h)^l (X_{gj}/h)^m.
\end{align*}
Furthermore, for $l\in \mathbb{N}_0$ and $\star, \diamond \in \{+,-\} $, let $\bar\mu^\star_l = \int k^\star(v)v^ldv$ and $\bar\kappa^\star_l = \int (k^\star(v))^2 v^l dv$.

\begin{lemma}\label{lemma::snj}
	Suppose that Assumption~\ref{ass:distX} holds, and either
	\begin{itemize}
		\item [(A)] Assumption~\ref{ass:LLI_clustersize}(i) holds and $\displaystyle \frac{1}{n^2} \sum_{g \in [G]} n_g^2  = o(1)$; or
		\item [(B)]  $\displaystyle \frac{1}{n^2} \sum_{g \in [G]} n_g^2  = o(h)$. 
	\end{itemize}
	Then the following hold for $l,m \in \mathbb{N}_0$ and $\star, \diamond \in \{+,-\} $.
	\begin{enumerate}[label=(\roman*)]
		\item $ S^\star_{n,l} = \bar\mu^\star_l f(0) + o_P(1),$
		\item $ \displaystyle T^\star_{n,l} = \frac{1}{nh} \left(\bar\kappa_l^\star f(0) + o_P(1)\right).$
		\item In case~(A),
		\begin{align*}
			nh U^{\star\diamond}_{n,l,m} 
			& = O_P\left( \frac{1}{nh}\sum_{g \in [G]}(n_gh)^2 + \frac{1}{nh} \right). 
		\end{align*}
		
		If in addition  all pairs $(X_{gi}, X_{gj})$, $i\neq j$, are identically distributed with continuous joint density $f(x_1,x_2)$ and
		$\displaystyle \frac{ \max_{g \in [G]} (n_gh)^2 }{nh + \sum_{g \in [G]} (n_gh)^2} = o(1)$,
		then 
		$$	nh U^{\star\diamond}_{n,l,m} = \lambda_n \bar\mu^\star_l \bar\mu^\diamond_m  f(0^{\star},0^{\diamond})+ o_P(1 + \lambda_n).$$
		
		\item In case~(B), 
		\begin{align*}
			nh U^{\star\diamond}_{n,l,m} & =  O_P\left( \frac{1}{n} \sum_{g \in [G]} n_g^2 + \sqrt{\frac{\max_{g \in [G]} n_g^2 }{nh}  \frac{1}{n}\sum_{g \in [G]} n_g^2} \right).
		\end{align*}	
		If in addition the realizations of the running variable are equal within each cluster and
		$\displaystyle\frac{\max_{g \in [G]} n_g^2}{\sum_{g \in [G]} n_g^2 } = o(h)$,
		then
		\[
		nh U^{\star\diamond}_{n,l,m}  = \bar\kappa^\star_{l+m} f(0) \lambda_n/h + o_P(\lambda_n/h).
		\]
		
	\end{enumerate}
\end{lemma}

The second lemma is relevant for the proofs under Asymptotic Frameworks~I and~III, where we assume that the density of the running variable within each cluster admits a bounded density.
Let $n_{g,h}$ denote the number of observations within the estimation window from cluster $g \in [G]$, i.e., $n_{g,h} = \sum_{i \in I_g} \1{|X_{gi}| \leq h}$, assuming that the support of the kernel function used is contained in $[-1,1]$.

\begin{lemma}\label{lemma:ngh}
	Suppose that Assumption~\ref{ass:LLI_clustersize}(i) holds, the kernel $k$ has bounded support, and $G \ge 2$. Then 
	\[
	\max_{g \in [G]} n_{g,h} = O_P\left( \max_{g \in [G]} n_g h + \log G \right).
	\]		
\end{lemma}

\subsection{Proof of Proposition~\ref{prop:AF_I-II}}
In the following, let $C$ denote a generic positive constant that might differ between equations.			
Recall from Appendix~\ref{A:sec:weights} that for $\star \in \{+,-\}$, the local linear weights can be expressed as
\[
w^\star_{gi}(h) = \frac{ \frac{1}{n} k^\star_h(X_{gi}) \left(S^\star_{n,2}   - S^\star_{n,1}  (X_{gi}/h) \right)  }{ S^\star_{n,2} S^\star_{n,0}- (S^\star_{n,1})^2 }.
\]
By Lemma~\ref{lemma::snj},
$S^+_{n,2} S^+_{n,0}- (S^+_{n,1})^2 = S^-_{n,2} S^-_{n,0}- (S^-_{n,1})^2 + o_P(1) = C + o_P(1)$. It follows that
\[
	\wt w^\star_{gi}(h) \equiv \frac{1}{n} k^\star_h(X_{gi}) \left(S^\star_{n,2}   - S^\star_{n,1}  (X_{gi}/h) \right) = w^\star_{gi}(h)/(C + o_P(1)).
\]
Let $\wt w_{gi}(h) = \wt w^+_{gi}(h) - \wt w^-_{gi}(h)$.

\medskip 

\noindent\textbf{Verification of Assumption~\ref{ass:HL_cluster_size}:} To begin with, note that the assumption that the eigenvalues of $\Sigma_g$ are bounded away from zero implies that
\[
se^2(h) = \frac{1}{C+o_P(1)} \sum_{g \in [G]} \sum_{i,j \in I_g} \wt w_{gi}(h) \wt w_{gj}(h) \sigma_{g, ij}  \ge \frac{1}{C+o_P(1)}\sum_{g \in [G]} \sum_{i \in I_g} \wt w_{gi}(h)^2.
\]
It follows that for any $g \in [G]$,
\begin{align}\label{eq:bound1}
	\frac{	\sum_{i,j \in I_g} |w_{gi}(h) w_{gj}(h)|}{se^2(h)} \leq  
	\frac{ \sum_{i,j \in I_g} |\wt w_{gi}(h) \wt w_{gj}(h)| }{ \sum_{g \in [G]} \sum_{i \in I_g} \wt w_{gi}(h)^2}(C + o_P(1)).
\end{align}

We use the bound in \eqref{eq:bound1} to verify Assumption~\ref{ass:HL_cluster_size}.
First, we note that the denominator of the bound satisfies
\begin{equation}\label{eq:step1}
	\sum_{g \in [G]} \sum_{i \in I_g} \wt w_{gi}(h)^2 = \sum_{\star \in \{+,-\}} \sum_{g \in [G]} \sum_{i \in I_g} \wt w^\star_{gi}(h)^2 = \frac{C+ o_P(1)}{nh}.
\end{equation}
This holds because by Lemma~\ref{lemma::snj}, for $\star \in \{+,-\}$, we have that
\begin{align*}
	\sum_{g \in [G]} \sum_{i \in I_g} \wt w^\star_{gi}(h)^2
	& =	\frac{1}{n^2}\sum_{g \in [G]} \sum_{i \in I_g} (k^\star_h(X_{gi}))^2 \left( (S^\star_{n,2})^2 - 2 S^\star_{n,2} S^\star_{n,1}  (X_{gi}/h) + (S^\star_{n,1}  (X_{gi}/h))^2 \right) \\
	& = (S^\star_{n,2})^2  T^\star_{n,0} - 2 S^\star_{n,2} S^\star_{n,1} T^\star_{n,1} + (S^\star_{n,1})^2  T^\star_{n,2}  \\
	& = \frac{1}{nh}\left( (\bar\mu_2^\star)^2 \bar\kappa_0^\star - 2 \bar\mu_2^\star \bar\mu_1^\star \bar\kappa_1^\star + (\bar\mu_1^\star)^2 \bar\kappa_2^\star     \right) f(0)^3 (1 + o_P(1)),
\end{align*}
where $(\bar\mu_2^\star)^2 \bar\kappa_0^\star - 2 \bar\mu_2^\star \bar\mu_1^\star \bar\kappa_1^\star + (\bar\mu_1^\star)^2 \bar\kappa_2^\star = \int (k^\star(v)(\bar\mu_2^\star - \bar\mu^\star_1 v ))^2 dv   > 0$.

Second, it holds that
\begin{align*}
	\max_{g \in [G]} \sum_{i,j \in I_g} |\wt w_{gi}(h) \wt w_{gj}(h)| 
	& \leq  \max_{g \in [G]} \max_{i \in I_g}  n_{g,h}^2 \wt w_{gi}(h)^2  = O_P \bigg( \frac{1 }{n^2h^2} \bigg) \max_{g \in [G]}n_{g,h}^2.
\end{align*}
By Assumption~\ref{ass:LLI_clustersize}, using Lemma~\ref{lemma:ngh}, or directly by Assumption~\ref{ass:LLII_clustersize} it then follows that
\begin{equation}\label{eq:step2}
	\max_{g \in [G]} \sum_{i,j \in I_g} |\wt w_{gi}(h) \wt w_{gj}(h)| = o_P\left( \frac{1}{nh} \right).
\end{equation}

Third, we show that
\begin{equation}\label{eq:step3}
	\sum_{g \in [G]}	\sum_{i,j \in I_g} |\wt w_{gi}(h) \wt w_{gj}(h)| = \sum_{\star, \diamond \in \{+, -\}} \sum_{g \in [G]}	\sum_{i,j \in I_g} |\wt w_{gi}^\star(h) \wt w^\diamond_{gj}(h)| = O_P\left(\frac{1}{nh}\right).
\end{equation}	
To prove this claim, note that
\begin{align}
	\sum_{g \in [G]}& 	\sum_{i,j \in I_g} |\wt w_{gi}^\star(h) \wt w_{gj}^\diamond(h)| \nonumber\\
	& \leq 	\frac{1}{n^2} \sum_{g \in [G]}	\sum_{i,j \in I_g} k^\star_h(X_{gi}) k^\diamond_h(X_{gj})  \left|  \big(S^\star_{n,2}   - S^\star_{n,1}  (X_{gi}/h) \big) \big(S^\diamond_{n,2}   - S^\diamond_{n,1}  (X_{gj}/h) \big) \right| \nonumber\\
	& \leq   S^\star_{n,2}S^\diamond_{n,2}  V^{\star\diamond}_{n,0,0} +  |S^\star_{n,2} S^\diamond_{n,1}  V^{\star\diamond}_{n,0,1}|  +  |S^\star_{n,1} S^\diamond_{n,2}   V^{\star\diamond}_{n,1,0}|   +  |S^\star_{n,1}S^\diamond_{n,1} V^{\star\diamond}_{n,1,1}| \label{eq:bound_V},
\end{align}
where $$	V^{\star\diamond}_{n,l,m} =  \frac{1}{n^2} \sum_{g \in [G]} \sum_{ i,j \in I_g} k^{\star}_h(X_{gi})k^{\diamond}_h(X_{gj})(X_{gi}/h)^l (X_{gj}/h)^m  = U^{\star\diamond}_{n,l,m} + T^\star_{n,l} \1{\star=\diamond }.$$
Under the assumptions of the proposition, $V^{\star\diamond}_{n,l,m} = O_P((nh)^{-1})$ by Lemma~\ref{lemma::snj}, and the conclusion in \eqref{eq:step3} follows.

Combining steps \eqref{eq:bound1}--\eqref{eq:step3}, both conditions of Assumption~\ref{ass:HL_cluster_size} follow:
\begin{align*}
	\max_{g \in [G]} \sum_{i,j \in I_g} \frac{ |w_{gi}(h) w_{gj}(h)|}{se^2(h)} & \leq (C + o_P(1))  \frac{ \max_{g \in [G]} \sum_{i,j \in I_g} |\wt w_{gi}(h) \wt w_{gj}(h)|}{\sum_{g \in [G]} \sum_{i \in I_g} \wt w_{gi}(h)^2}   =o_p(1),\\	
	\sum_{g \in [G]} \sum_{i,j \in I_g} \frac{ |w_{gi}(h) w_{gj}(h)|}{se^2(h)} & \leq (C + o_P(1)) \frac{ \sum_{g \in [G]}  \sum_{i,j \in I_g} |\wt w_{gi}(h) \wt w_{gj}(h)|}{\sum_{g \in [G]} \sum_{i \in I_g} \wt w_{gi}(h)^2}  = O_P(1).
\end{align*}

\medskip 

\noindent\textbf{Rate of the conditional variance:} 
We have already shown that $$se^2(h) \ge (C+o_P(1))/(nh).$$ Moreover,
\[
	se^2(h) \leq (C + o_P(1))	\sum_{g \in [G]}	\sum_{i,j \in I_g} |\wt w_{gi}(h) \wt w_{gj}(h)| =  O_P\left(\frac{1}{nh}\right).
\]
where the first step uses the assumption that the conditional variances (and hence also covariances) are bounded and the second step uses~\eqref{eq:step3}.
Together, these results imply that $se^2(h) \asymp_p (nh)^{-1}$. 

\medskip 

\noindent\textbf{Limit of the conditional worst-case bias:} Note that
\begin{equation}\label{eq:proof_bias}
	\sum_{g \in [G]} \sum_{i \in I_g} w^\star_{gi}(h) X_{gi}^2 = \frac{  (S^\star_{n,2})^2  - S^\star_{n,1}S^\star_{n,3}  }{ S^\star_{n,2} S^\star_{n,0}- (S^\star_{n,1})^2 } h^2 = (\bar\mu + o_P(1))h^2.
\end{equation}
where the second step follows by Lemma~\ref{lemma::snj}. \qed

\subsection{Proof of Proposition~\ref{prop:AF_III-IV}}
The proof has a similar structure as the proof of Proposition~\ref{prop:AF_I-II}, except that we directly use the assumed order of the conditional variance, rather than derive it. To begin with, we note that under either Assumption~\ref{ass:LL3_clustersize} or~\ref{ass:LL4_clustersize}, Lemma~\ref{lemma::snj} yields $S^\star_{n,l} = \bar\mu^\star_lf(0) + o_P(1)$ for $\star \in \{+,-\}$ and $l \in \mathbb{N}_0$. It follows that
$S^+_{n,2} S^+_{n,0}- (S^+_{n,1})^2 = S^-_{n,2} S^-_{n,0}- (S^-_{n,1})^2 + o_P(1) = C + o_P(1)$ for a positive constant $C$, and the limit of the worst-case conditional bias follows as in \eqref{eq:proof_bias}.

We verify Assumption~\ref{ass:HL_cluster_size} in two steps; first under Assumption~\ref{ass:LL3_clustersize} and then under Assumption~\ref{ass:LL4_clustersize}

\medskip

\noindent\textbf{Part 1:} 
Suppose that Assumption~\ref{ass:LL3_clustersize} holds.
First, by Lemma~\ref{lemma:ngh},
\[
\max_{g \in [G]} \sum_{i,j \in I_g} |w_{gi}(h)  w_{gj}(h)| = O_P\left( \frac{1}{n^2h^2} \right) \max_{g \in [G]} n_{g,h}^2 = O_P\left( \frac{\max_{g \in [G]} n_{g}^2h^2 + \log^2 G}{n^2h^2} \right).
\]
Second, using the inequality~\eqref{eq:bound_V}, Lemma~\ref{lemma::snj} yields 
\begin{align*}
	\sum_{g \in [G]} 	\sum_{i,j \in I_g} |w_{gi}(h) w_{gj}(h)| = O_P\left(\frac{1 + \lambda_n}{nh}\right).
\end{align*}
Since by assumption $se^2(h) \asymp_P (1 + \lambda_n)/(nh)$, Assumption~\ref{ass:HL_cluster_size} follows under the assumptions made.

\medskip

\noindent\textbf{Part 2:} Now suppose that Assumption~\ref{ass:LL4_clustersize} holds. First, note that
\[
\max_{g \in [G]} \sum_{i,j \in I_g} | w_{gi}(h) w_{gj}(h)| = O_P\left( \frac{\max_{g \in [G]} n_g^2}{n^2h^2} \right).
\]
Since by assumption $se^2(h) \asymp_P (1 + \lambda_n/h)/(nh) =  \sum_{g \in [G]}n_g^2/(n^2h)$, we obtain that
\begin{align*}
	\frac{\max_{g \in [G]} \sum_{i,j \in I_g} | w_{gi}(h) w_{gj}(h)|}{se^2(h)}  =  O_P\left( \frac{\max_{g \in [G]} n_{g}^2}{ h \sum_{g \in [G]} n_g^2}  \right).
\end{align*}
Part~(i) of Assumption~\ref{ass:HL_cluster_size} follows.

Second, using the inequality~\eqref{eq:bound_V}, Lemma~\ref{lemma::snj} yields 
\begin{align*}
	\sum_{g \in [G]} 	\sum_{i,j \in I_g} |w_{gi}(h) w_{gj}(h)| =  O_P\left(\frac{1 + \lambda_n/h}{nh}\right).
\end{align*}
This shows that part~(ii) of Assumption~\ref{ass:HL_cluster_size} is satisfied.\qed

\subsection{Proofs of Lemmas~\ref{lemma:se_1} and~\ref{lemma:se_2}}			
\textbf{Part (i):} Under Assumption~\ref{ass:epsilon}(i), by the triangle inequality,
\[
	se^2(h) \leq C \sum_{g \in [G]}	\sum_{i,j \in I_g} |w_{gi}(h) w_{gj}(h)|
\]
for a positive constant $C$.
The conclusions follow using the inequality~\eqref{eq:bound_V} and Lemma~\ref{lemma::snj}.

\medskip

\noindent\textbf{Part (ii):}
Under Assumption~\ref{ass:limit_variance}, we have that
\begin{align*}
	se^2(h) & =   \sum_{g \in [G]} \sum_{i \in I_g}  w_{gi}(h)^2 \sigma_{g,i}^2 +  \sum_{g \in [G]}  \sum_{\substack{ i \neq j  \\ i,j \in I_g} } w_{gi}(h) w_{gj}(h) \sigma_{g,ij} \\
	& =   \sum_{\star \in \{+,-\}} \sigma^2(0^\star) \sum_{g \in [G]} \sum_{i \in I_g}  w_{gi}^\star(h)^2 (1+ o_P(1)) \\
	& \quad +   \sum_{\star, \diamond \in \{+, -\}} \onepm{\star = \diamond} \sigma(0^\star,0^\diamond) \sum_{g \in [G]}  \sum_{\substack{ i \neq j  \\ i,j \in I_g} } w^\star_{gi}(h) w^\diamond_{gj}(h) (1+ o_P(1)).
\end{align*}
Recall that
$w^\star_{gi}(h) = \frac{1}{n} k^\star_h(X_{gi}) \left(S^\star_{n,2}   - S^\star_{n,1}  (X_{gi}/h) \right)  / ( S^\star_{n,2} S^\star_{n,0}- (S^\star_{n,1})^2 ).$	
Using Lemma~\ref{lemma::snj}, we obtain that 
\[
	\sum_{g \in [G]} \sum_{i \in I_g} w^\star_{gi}(h)^2 = \frac{1}{nh} \frac{\bar\kappa}{f(0)}(1+o_P(1)).
\]
The second component satisfies:
\begin{align*}
		\sum_{g \in [G]}  \sum_{\substack{ i \neq j  \\ i,j \in I_g} } w^\star_{gi}(h) w^\diamond_{gj}(h)  
		& = \frac{
		\bar{\mu}^\star_{2}\bar{\mu}^\diamond_{2} U^{\star\diamond}_{n,0,0}
		- \bar{\mu}^\star_{2}\bar{\mu}^\diamond_{1} U^{\star\diamond}_{n,0,1}
		- \bar{\mu}^\star_{1}\bar{\mu}^\diamond_{2} U^{\star\diamond}_{n,1,0}
		+ \bar{\mu}^\star_{1}\bar{\mu}^\diamond_{1} U^{\star\diamond}_{n,1,1}
		}{\left(\bar{\mu}^\star_2 \bar{\mu}^\star_0 - (\bar{\mu}^\star_1)^2\right) \left(\bar{\mu}^\diamond_2 \bar{\mu}^\diamond_0 - (\bar{\mu}^\diamond_1)^2\right)  f^2_X(0) } (1 + o_P(1)).
\end{align*}
Under the assumptions of Lemma~\ref{lemma:se_1}, plugging in $	nh U^{\star\diamond}_{n,l,m}  = \lambda_n \bar\mu_l^\star\bar\mu_m^\diamond f(0,0) + o_P(1 + \lambda_n)$, we obtain that
\begin{align*}
	\sum_{g \in [G]}  \sum_{\substack{ i \neq j  \\ i,j \in I_g} } w^\star_{gi}(h) w^\diamond_{gj}(h)  = 
	\frac{1}{nh} \left( \frac{f(0,0)}{f^2_X(0)} \lambda_n + o_p(1 + \lambda_n) \right).
\end{align*}
Under the assumptions of Lemma~\ref{lemma:se_2}, plugging in $	nh U^{\star\diamond}_{n,l,m}  = \lambda_n \bar\kappa^\star_{l+m}f_X(0)/h + o_P(1 + \lambda_n/h)$, we obtain that
\begin{align*}
	\sum_{g \in [G]}  \sum_{\substack{ i \neq j  \\ i,j \in I_g} } w^\star_{gi}(h) w^\diamond_{gj}(h)  = 
	\frac{1}{nh} \left( \frac{\bar\kappa}{f_X(0)} \frac{\lambda_n}{h} + o_P\left(1 + \frac{\lambda_n}{h}\right) \right) \1{\star = \diamond}.
\end{align*}
Part~(ii) of both lemmas follows.\qed

\section{Proofs of Auxiliary Lemmas}\label{sec:proofs_lemmas}	

\subsection{Proof of Lemma~\ref{lemma::snj}}
In case~(A), for any $g \in [G]$ and $i,j \in I_g$, $i \neq j$, let $f_{X_{gi},X_{gj}}$ denote the joint density of $X_{gi}$ and $X_{gj}$. Note that it is uniformly bounded under Assumption~\ref{ass:LLI_clustersize}(i).

\noindent\textbf{Part~(i).} We study the expectation and the variance of $S^\star_{n,l}$. First, note that
\[
\E[S^\star_{n,l}] = \E[k^\star_h(X_{gi})(X_{gi}/h)^l] = f(0)\bar\mu^\star_l + o(1). 
\]
Second, it holds that
\begin{align*}
	\Var(S^\star_{n,l} ) 
	& =  \frac{1}{(nh)^2}\sum_{g \in [G]} \Var\left( \sum_{i \in I_g} k^\star(X_{gi}/h)(X_{gi}/h)^l \right) \\
	& = \frac{1}{(nh)^2}\sum_{g \in [G]} n_g \Var( k^\star(X_{g1}/h)(X_{g1}/h)^l ) \\
	& \quad + \frac{1}{(nh)^2}\sum_{g \in [G]}  \sum_{\substack{ i \neq j  \\ i,j \in I_g} }  \Cov( k^\star(X_{gi}/h)(X_{gi}/h)^l,\,k^\star(X_{gj}/h)(X_{gj}/h)^l ) 
\end{align*}

We analyze this variance in cases (A) and (B) separately, using standard kernel calculations.
In Case~(A), 
\begin{align*}
	\Var(S^\star_{n,j}) 
	&  \leq \frac{1}{(nh)^2}\sum_{g \in [G]} n_g  \int (k^\star(x/h)(x/h)^l)^2 f_X(x)dx \\
	& \quad + \frac{1}{(nh)^2}\sum_{g \in [G]} \sum_{\substack{ i \neq j  \\ i,j \in I_g} }   \left( \int \int k^\star(x/h)(x/h)^l k^\star(z/h)(z/h)^l  f_{X_{gi},X_{gj}}(x,z)dxdz \right) \\
	& = \frac{1}{n^2 h}\sum_{g \in [G]} n_g \int (k^\star(v)(v)^l)^2 f(vh)dv \\
	& \quad + \frac{1}{n^2 }\sum_{g \in [G]} \sum_{\substack{ i \neq j  \\ i,j \in I_g} } \int \int k^\star(v)v^l k^\star(w)w^l f_{X_{gi},X_{gj}}(vh, wh)dvdw \\
	& = O\left( \frac{1}{n h} +  \frac{\sum_{g \in [G]} n_g(n_g-1)}{n^2} \right) = o(1).
\end{align*}	
In Case~(B), we employ the inequality
\[
	\Cov( k^\star(X_{gi}/h)(X_{gi}/h)^l,\,k^\star(X_{gj}/h)(X_{gj}/h)^l )  \leq \Var(k^\star(X_{g1}/h)(X_{g1}/h)^l ).
\]
With this bound, we obtain that
\begin{align*}
	\Var(S^\star_{n,j}) & \leq \frac{1}{(nh)^2}\sum_{g \in [G]} n_g^2 \Var( k^\star(X_{g1}/h)(X_{g1}/h)^l ) \\
	&  \leq \frac{1}{(nh)^2}\sum_{g \in [G]} n_g^2  \int (k^\star(x/h)(x/h)^l)^2 f(x)dx \\
	& =  \frac{1}{n^2 h}\sum_{g \in [G]} n_g^2 \int (k^\star(v)(v)^l)^2 f(vh)dv \\
	& = O\left(\frac{\sum_{g \in [G]} n_g^2 }{n^2 h}\right) = o(1).
\end{align*}	
This concludes the proof of part~(i).

\noindent\textbf{Part~(ii).} We study the expectation and the variance of $T^\star_{n,l}$. First, note that
\begin{align*}
	\E[T^\star_{n,l}] & = \frac{1}{n}\E[(k^\star_h(X_{gi}))^2(X_{gi}/h)^l] \\
	& = \frac{1}{n} \int (k^\star_h(x))^2(x/h)^l f(x)dx \\
	& =  \frac{1}{nh}(\bar\kappa^\star_l f(0) + o(1)).
\end{align*}
Second, it holds that
\begin{align*}
	\Var(T^\star_{n,l}) & =  \Var\left( \frac{1}{n^2h^2}\sum_{g \in [G]} \sum_{i \in I_g} k^\star(X_{gi}/h)^2(X_{gi}/h)^{l}  \right) \\
	& = \frac{1}{(nh)^4}\sum_{g \in [G]} n_g \Var( (k^\star(X_{g1}/h))^2(X_{g1}/h)^l ) \\
	& \quad + \frac{1}{(nh)^4}\sum_{g \in [G]}\sum_{\substack{ i \neq j  \\ i,j \in I_g} } \Cov( (k^\star(X_{gi}/h))^2 (X_{gi}/h)^l,\,(k^\star(X_{gj}/h))^2(X_{gj}/h)^l  ).
\end{align*}
We study this variance separately in cases (A) and (B), using standard kernel derivations.
In Case~(A),
\begin{align*}
	\Var(nhT^\star_{n,j}) 
	&  \leq \frac{1}{(nh)^2}\sum_{g \in [G]} n_g  \int ((k^\star(x/h))^2(x/h)^l)^2 f(x)dx \\
	& \quad + \frac{1}{(nh)^2}\sum_{g \in [G]} \sum_{\substack{ i \neq j  \\ i,j \in I_g} }  \left( \int \int (k^\star(x/h))^2(x/h)^l (k^\star(z/h))^2(z/h)^l  f(x,z)dxdz \right) \\
	& = \frac{1}{n^2 h}\sum_{g \in [G]} n_g \int ((k^\star(v))^2(v)^l)^2 f(vh)dv \\
	& \quad + \frac{1}{n^2 }\sum_{g \in [G]}  \sum_{\substack{ i \neq j  \\ i,j \in I_g} } \int \int (k^\star(v))^2v^l (k^\star(w))^2w^l f(vh, wh)dvdw \\
	& = O\left( \frac{1}{nh} +  \frac{\sum_{g \in [G]} n_g(n_g-1)}{n^2 } \right) = o(1).
\end{align*}
In Case~(B),
\begin{align*}
	\Var(nh T^\star_{n,l})	& \leq \frac{1}{(nh)^2}\sum_{g \in [G]} n_g^2 \Var( k^\star(X_{g1}/h)^2(X_{g1}/h)^l  ) \\
	&  \leq \frac{1}{(nh)^2}\sum_{g \in [G]} n_g^2  \int (k^\star(x/h)^2(x/h)^l)^2 f(x)dx \\
	& =  \frac{1}{n^2 h}\sum_{g \in [G]} n_g^2 \int (k^\star(v)^2(v)^l)^2 f(vh)dv \\
	& = O\left( \frac{\sum_{g \in [G]} n_g^2 }{n^2 h} \right) =o(1),
\end{align*}	
where the first inequality follows by the Cauchy-Schwarz inequality for the covariances. This concludes the proof of part~(ii).

\noindent\textbf{Part~(iii).}
We study the expectation and the variance of $U^{\star\diamond}_{n,l,m}$ in Case~(A). First, by standard kernel derivations,
\begin{align*}
	\E[U^{\star\diamond}_{n,l,m}] &= \frac{1}{n^2} \sum_{g \in [G]}  \sum_{\substack{ i \neq j  \\ i,j \in I_g} } \E[k^{\star}_h(X_{gi})(X_{gi}/h)^l k^{\diamond}_h(X_{gj})(X_{gj}/h)^m] \\
	& = \frac{1}{n^2} \sum_{g \in [G]} \sum_{\substack{ i \neq j  \\ i,j \in I_g} } \int\int k^{\star}_h(x_1)(x_1/h)^l k^{\diamond}_h(x_2)(x_2/h)^m f_{X_{gi},X_{gj}}(x_1,x_2)dx_1dx_2 \\
	& = \frac{1}{n^2} \sum_{g \in [G]}\sum_{\substack{ i \neq j  \\ i,j \in I_g} } \int\int k^{\star}(v)v^l k^{\diamond}(w)(w)^m f_{X_{gi},X_{gj}}(vh,wh)dvdw \\
	& = O\left( \frac{1}{n^2}\sum_{g \in [G]} n_g(n_g-1) \right).
\end{align*}

Next, we consider the variance. Let
\begin{align*}
	C_g(i_1,j_1,i_2,j_2) = \Cov\Big(& k^{\star}_h(X_{gi_1})k^{\diamond}_h(X_{gj_1})(X_{gi_1}/h)^l (X_{gj_1}/h)^m, \\
& \quad k^{\star}_h(X_{gi_2})k^{\diamond}_h(X_{gj_2})(X_{gi_2}/h)^l (X_{gj_2}/h)^m \Big).
\end{align*}
Note that for all indices $i_1,j_1,i_2,j_2 \in I_g$ such that $i_1\neq j_1$ and $i_2 \neq j_2$,
\[
|C_g(i_1,j_1,i_2,j_2)| \leq \begin{cases}
	C/h^2 & \text{ if } i_1=i_2 \text{ and } j_1 = j_2, \\
	C/h	& \text{ if there is exactly one pair of equal indices,}\\
	C	& \text{ if } i_1, i_2, j_1, j_2 \text{ are pairwise different.} \\		
\end{cases}
\]
for some constant $C$. We have that
\begin{align*}
	\Var(U^{\star\diamond}_{n,l,m})
	&= \frac{1}{n^4} \sum_{g \in [G]} \Var\left(   \sum_{\substack{ i \neq j  \\ i,j \in I_g} } k^{\star}_h(X_{gi})k^{\diamond}_h(X_{gj})(X_{gi}/h)^l (X_{gj}/h)^m \right) \\
	&= \frac{1}{n^4} \sum_{g \in [G]}  \sum_{\substack{ i \neq j   \\ i,j \in I_g} } C_g(i,j,i,j)
	+ \frac{1}{n^4} \sum_{g \in [G]}  \sum_{\substack{ i_1,j_1,i_2,j_2 \in I_g \\ \text{pairwise different}} } 	C_g(i_1,j_1,i_2,j_2) \\
	& \quad + \frac{1}{n^4} \sum_{g \in [G]}  \sum_{\substack{ i_1 \neq j_1, i_2 \neq j_2 \\ i_1 = i_2, j_1 \neq j_2  \\ i_1,j_1,i_2,j_2 \in I_g} } 	C_g(i_1,j_1,i_2,j_2) 
	+ \frac{1}{n^4} \sum_{g \in [G]}  \sum_{\substack{ i_1 \neq j_1, i_2 \neq j_2 \\ i_1 \neq i_2, j_1 = j_2  \\ i_1,j_1,i_2,j_2 \in I_g} } 	C_g(i_1,j_1,i_2,j_2) \\
	& \quad + \frac{1}{n^4} \sum_{g \in [G]}  \sum_{\substack{ i_1 \neq j_1, i_2 \neq j_2 \\ i_1 = j_2, i_2 \neq j_1  \\ i_1,j_1,i_2,j_2 \in I_g} } 	C_g(i_1,j_1,i_2,j_2) 
	+ \frac{1}{n^4} \sum_{g \in [G]}  \sum_{\substack{ i_1 \neq j_1, i_2 \neq j_2 \\ i_2 = j_1, i_1 \neq j_2  \\ i_1,j_1,i_2,j_2 \in I_g} } 	C_g(i_1,j_1,i_2,j_2) \\
	& = O\left(\frac{1}{n^4h^2}\sum_{g \in [G]}n_g^2 + \frac{1}{n^4 h}\sum_{g \in [G]} n_g^3 + \frac{1}{n^4}\sum_{g \in [G]}n_g^4 
	\right) \\
	& = O\left(\frac{1}{(nh)^4}\sum_{g \in [G]}(n_gh)^2 +  \frac{1}{(nh)^4}\sum_{g \in [G]}(n_gh)^4
	\right),
\end{align*}
where the last step uses the Cauchy-Schwarz inequality. The first statement of part~(iii) follows by noting that
\begin{align*}
	nh U^{\star\diamond}_{n,l,m} & = O_P\left( \frac{1}{nh}\sum_{g \in [G]}(n_gh)^2 + \frac{1}{nh}\sqrt{\sum_{g \in [G]}(n_gh)^2} +  \frac{1}{nh} \sqrt{\sum_{g \in [G]}(n_gh)^4} \right) \\
	& =  O_P\left( \frac{1}{nh}\sum_{g \in [G]}(n_gh)^2 + \frac{1}{nh}\left(1 + \sum_{g \in [G]}(n_gh)^2 \right) + \frac{1}{nh} \sqrt{ \left(\sum_{g \in [G]}(n_gh)^2\right)^2} \right) \\
	& = O_P\left( \frac{1}{nh}\sum_{g \in [G]}(n_gh)^2 + \frac{1}{nh} \right). 
\end{align*}

To prove the second claim, note that if all pairs $(X_{gi}, X_{gj})$, $i\neq j$, are identically distributed with continuous joint density $f(x_1,x_2)$, then
\begin{align*}
	\E[nh U^{\star\diamond}_{n,l,m}]	& = \frac{nh}{n^2} \sum_{g \in [G]} n_g(n_g-1) \int\int k^{\star}_h(x_1)(x_1/h)^l k^{\diamond}_h(x_2)(x_2/h)^m f(x_1,x_2)dx_1dx_2\\
	& = \lambda_n \int\int k^{\star}(v)v^l k^{\diamond}(w)(w)^m f(vh,wh)dvdw \\
	& = \lambda_n \int\int k^{\star}(v)v^l k^{\diamond}(w)(w)^m dvdw f(0^{\star},0^{\diamond}) + o(1+o(\lambda_n)) \\
	& = \lambda_n \bar\mu^\star_l \bar\mu^\diamond_m  f(0^{\star},0^{\diamond}) + o(1+o(\lambda_n)).
\end{align*}
If in addition,
\[
\frac{\max_{g \in [G]} (n_gh)^2}{nh + \sum_{g \in [G]} (n_gh)^2} = o(1),
\]
then we obtain that
\begin{align*}
	nh U^{\star\diamond}_{n,l,m} & = \lambda_n \bar\mu^\star_l \bar\mu^\diamond_m  f(0^{\star},0^{\diamond})(1+o(1))
	+ O_P\left( \sqrt{ \frac{1}{(nh)^2}\sum_{g \in [G]}(n_gh)^2\left(1 + \max_{g \in [G]} (n_gh)^2\right)} \right) \\
	& = \lambda_n \bar\mu^\star_l \bar\mu^\diamond_m  f(0^{\star},0^{\diamond})(1+o(1)) \\
	&  \quad + O_P\left(\sqrt{ \frac{ (nh + \sum_{g \in [G]} (n_gh)^2) \sum_{g \in [G]} (n_gh)^2}{(nh)^2} \frac{1 + \max_{g \in [G]} (n_gh)^2}{nh + \sum_{g \in [G]} (n_gh)^2}}\right) \\
	& = \lambda_n \bar\mu^\star_l \bar\mu^\diamond_m  f(0^{\star},0^{\diamond}) (1 + o_P(1)) + o_P(1).
\end{align*}
This concludes the proof of part~(iii).

\noindent\textbf{Part~(iv).}
We study the expectation and the variance of $U^{\star\diamond}_{n,l,m}$ in Case~(B). First, note that
\begin{align*}
	|\E[U^{\star\diamond}_{n,l,m}] | & \leq \frac{1}{n^2} \sum_{g \in [G]}  \sum_{\substack{ i \neq j  \\ i,j \in I_g} }  \E[k^{\star}_h(X_{gi}) k^{\diamond}_h(X_{gj}) ] \\
	& \leq \frac{1}{n^2} \sum_{g \in [G]} n_g(n_g-1)  \E[(k^{\star}_h(X_{g1}))^2] \\
	& = O\left( \frac{1}{n^2h} \sum_{g \in [G]} n_g(n_g-1)  \right),
\end{align*}
where the second inequality follows by the Cauchy-Schwarz inequality. If all the realizations of the running variable are equal within each cluster and  $f(0)>0$, then
\begin{align*}
	\E[	U^{\star\diamond}_{n,l,m}] & = \frac{1}{n^2} \sum_{g \in [G]} n_g(n_g-1)  \E\big[(k^{\star}_h(X_{g1}))^2 (X_{g1}/h)^{l+m} \big] \\
	& = \frac{1}{n^2} \sum_{g \in [G]} n_g(n_g-1) \bar\kappa^*_{l+m}f(0)(1+ o(1)).
\end{align*}

The variance of $U^{\star\diamond}_{n,l,m}$ is bounded as follows:
\begin{align*}
	\Var(U^{\star\diamond}_{n,l,m}) & = 
	\frac{1}{n^4} \sum_{g \in [G]} \Var\left(   \sum_{\substack{ i \neq j  \\ i,j \in I_g} } k^{\star}_h(X_{gi})k^{\diamond}_h(X_{gj})(X_{gi}/h)^l (X_{gj}/h)^m \right) \\
	& = \frac{1}{n^4} \sum_{g \in [G]}  \sum_{\substack{ i_1 \neq j_1  \\ i_1,j_1 \in I_g} } \sum_{\substack{ i_2 \neq j_2  \\ i_2,j_2 \in I_g} } 	C_g(i_1,j_1,i_2,j_2) \\
	&  \leq \frac{1}{n^4} \sum_{g \in [G]} (n_g(n_g-1))^2 \E\big[k_h(X_{g1})^4\big] \\
	& = O\left(\frac{1}{n^4 h^3} \sum_{g \in [G]} n_g^4 \right) \\
	& = O\left(\frac{1}{n^2 h^2} \frac{1}{n} \sum_{g \in [G]} n_g^2 \frac{\max_{g \in [G]} n_g^2 }{nh}  \right),
\end{align*} 
where $C_g(i_1,j_1,i_2,j_2)$ denotes respective variances, as defined in the proof of part~(iii). Both statements in part~(iv) follow from the above observations. \qed

\subsection{Proof of Lemma~\ref{lemma:ngh}}
By the union bound and Chernoff's inequality, for any $B>0$ and $t>0$,
\begin{align*} 
	\Pr\left(\max_{g \in [G]} n_{g,h} > B\right) & \leq G \max_{g \in [G]} \Pr\left( n_{g,h} > B  \right)  \leq G \frac{\max_{g \in [G]}\E[e^{tn_{g,h}}]}{e^{tB}}.
\end{align*}
Next, for any $g \in [G]$,
\begin{align*}
	\E[e^{tn_{g,h}}] & = \sum_{k=0}^{n_g}  e^{tk} \Pr(n_{g,h} = k) \\
	& \leq \sum_{k=0}^{n_g}  e^{tk} \sum_{1 \leq i_1 < \ldots < i_k \leq n_g} \Pr\left( \bigcap_{l=1}^k \left\{|X_{gi_l}| \leq h \right\} \right) \\
	& \leq \sum_{k=0}^{n_g} e^{tk} {n_g \choose k}  C(2h)^k \\
	& = C ( 2he^t + 1)^{n_g} \\
	& \leq C e^{2 n_g h e^t},
\end{align*}
where the second step uses the union bound, the third step uses the assumption of bounded join densities, the fourth step uses the binomial formula, and the last step uses the bound $1+x \leq e^{x}$.

It follows that
\[
\Pr\left(\max_{g \in [G]} n_{g,h} > B\right) \leq C G \frac{e^{m h e^t}}{e^{tB}}.
\]
where $m = 2 \max_{g \in [G]} n_{g}$.

Letting $B = \frac{e^t}{t}( mh + \log G)$, we obtain
\begin{align*}
	\Pr\left(\max_{g \in [G]} n_{g,h} > \frac{e^t}{t}( mh + \log G) \right) &  \leq C \exp\left( \log G +  mh e^t - t \frac{e^t}{t} (m  h  + \log G ) \right) \\
	& =   C \exp\left( - (e^t - 1) \log G \right).
\end{align*}
The last bound can be made arbitrarily small by choosing $t$ large enough, which concludes the proof. \qed

\singlespacing
\bibliography{bib}

\end{document}